\documentclass{theoretics}
\title{A Discrete Analog of Tutte's Barycentric Embeddings on Surfaces}

\ThCSauthor[LIGM]{Éric Colin de Verdière}{eric.colin-de-verdiere@univ-eiffel.fr}
\ThCSaffil[LIGM]{LIGM, CNRS, Univ Gustave Eiffel, F-77454 Marne-la-Vallée, France}
\ThCSauthor[IUF]{Vincent Despré}{vincent.despre@loria.fr}
\ThCSaffil[IUF]{Institut universitaire de France (IUF), Université de Lorraine, CNRS, LORIA, France}
\ThCSauthor[ND]{Loïc Dubois}{ldubois@nd.edu}[0000-0002-0366-4959]
\ThCSaffil[ND]{University of Notre Dame, USA}

\ThCSthanks{%
    The conference version of this paper was published in SODA 2025~\cite{DBLP:conf/soda/VerdiereD025}. %
    The first and second authors were partly supported by grant ANR-25-CE40-0416 of the French National Research Agency (project SUGAR).
    The second author was also partly supported by grant ANR-23-CE48-0017 of the French National Research Agency ANR (project Abysm).
    This work was done while the third author was working at LIGM, CNRS, Univ Gustave Eiffel, F-77454 Marne-la-Vallée, France.
}
\ThCSshortnames{É.\ Colin de Verdière, V.\ Despré, L.\ Dubois}
\ThCSshorttitle{A Discrete Analog of Tutte's Barycentric Embeddings on Surfaces}
\ThCSyear{2026}
\ThCSarticlenum{11}
\ThCSreceived{Jul 21, 2025}
\ThCSrevised{Mar 9, 2026}
\ThCSaccepted{May 5, 2026}
\ThCSpublished{Jul 3, 2026}
\ThCSdoicreatedtrue
\ThCSkeywords{Computational topology; surface; graph; algorithm; embedding; Tutte's barycentric method; reducing triangulations.}

\newcommand\emphdef[1]{\emph{\textbf{#1}}}

\makeatletter
\newcommand\IfRestateTF{%
  \ifx\label\thmt@gobble@label 
    \expandafter\@firstoftwo
  \else
    \expandafter\@secondoftwo
  \fi
}
\makeatother
\newcommand{\RestateRemark}{\IfRestateTF{{\normalfont\bfseries (Restated) }}{}}

\begin{document}


\maketitle

\begin{abstract}
  Tutte's celebrated barycentric embedding theorem describes a natural way to build straight-line embeddings (crossing-free drawings) of a (3-connected) planar graph: map the vertices of the outer face to the vertices of a convex polygon, and ensure that each remaining vertex is in \emph{convex position}, namely, a barycenter with positive coefficients of its neighbors.  Actually computing an embedding then boils down to solving a system of linear equations.  A particularly appealing feature of this method is the flexibility given by the choice of the barycentric weights.  Generalizations of Tutte's theorem to surfaces of non-positive curvature are known, but due to their inherently continuous nature, they do not lead to an algorithm. In fact, it is not even clear what a discrete version would be, and our first contribution is to propose a natural setup for this purpose.

  In this paper, we introduce a purely discrete analog of Tutte's theorem for surfaces (with or without boundary) of non-positive curvature, based on the recent notion of \emph{reducing triangulations}.  We prove a Tutte theorem in this setting: every drawing homotopic to an embedding such that each vertex is \emph{harmonious} (a discrete analog of being in convex position) is a \emph{weak embedding} (arbitrarily close to an embedding).  We also provide a polynomial-time algorithm to make an input drawing harmonious without increasing the length of any edge, in a similar way as a drawing can be put in convex position without increasing the edge lengths.
\end{abstract}

\section{Introduction}

\paragraph{Tutte's barycentric method and its generalizations.}

One of the most basic problems in graph drawing is that of constructing straight-line embeddings (crossing-free drawings) of planar graphs~\cite{v-pslda-06}.  Tutte's celebrated barycentric embedding theorem~\cite{tutte1963draw} (1963) lies at the root of the wide class of force-directed drawing algorithms~\cite{k-fdda-06}.  It provides a remarkably simple method to build straight-line embeddings of a (simple, 3-connected) planar graph, assuming the knowledge of a facial cycle in some (topological) planar embedding: map the vertices of that cycle to the vertices of a convex polygon, in the same order, and locate the other vertices in such a way that, if the inner edges are replaced by springs, then this physical system is at its equilibrium.

Tutte's initial proof~\cite{tutte1963draw} has been revisited many times; see, in particular, Richter-Gebert~\cite[Section~12.2]{r-rsp-96}, Thomassen~\cite{thomassen2004tutte}, or Edelsbrunner~\cite[Section~I.4]{eh-cti-09}, thus leading to more insight.  As it turns out, it suffices that each inner vertex be a barycenter of its neighbors, with some arbitrary positive coefficients~\cite{f-psast-97}.  (The coefficients need not be symmetric, so this is more general than requiring an equilibrium in a spring system, even if the springs may have different rigidities.)  For this reason, Tutte drawings are called \emph{barycentric}, and they can be computed by solving a system of linear equations.  Equivalently, it suffices that each vertex be in \emph{convex position} with respect to its neighbors: disregarding some degenerate cases, every straight line containing an inner vertex~$v$ sees edges incident to $v$ on both sides.  (The 3-connectivity assumption  avoids overlaps.)

Tutte's method has been seminal not only for graph drawing.  Floater and Gotsman~\cite{fg-hmti-99} and Gotsman and Surazhsky~\cite{gs-gifpm-01} use it to build morphings between two straight-line embeddings of the same planar graph.  In a nutshell, one can interpolate between two embeddings in convex position by moving the barycentric coefficients from those of the initial embedding to those of the final embedding; Tutte's theorem guarantees that the drawing will be an embedding at every step.  Other applications include surface parameterization and approximation~\cite{f-psast-97}.

Of special interest to us is a particular generalization of Tutte's theorem to graphs drawn on a surface~$S$ without boundary and of non-positive curvature (in particular, such surfaces are not homeomorphic to a sphere).  Consider a graph~$G$ embedded as (the 1-skeleton of) a triangulation of~$S$ (the embedding needs not be geodesic).  Y.~Colin de Verdière~\cite[Théorème~2]{de1991comment} (see also Hass and Scott~\cite[Lemma~10.12]{hs-seshm-15} and Luo, Wu, and Zhu~\cite[Theorem~1.6]{lwz-dsgtg-23}) proves that if this embedding is deformed by a homotopy (a continuous motion) in such a way that, in the resulting drawing, each edge is geodesic and each vertex is in convex position, then this drawing is actually an embedding.  A similar result holds for surfaces with boundary~\cite[Théorème~3]{de1991comment}, fixing part of the triangulation to the boundary.  Without the requirement to have a triangulation, degenerate cases could occur, so one would not always get an embedding.  The assumption on non-positive curvature seems necessary because it ensures that geodesics (locally shortest paths) cross minimally.

As a side note, for the special case of the torus, other graph embedding methods based on Schnyder woods exist~\cite{gl-tmswo-14}, but they lack the flexibility of barycentric embeddings, and to our knowledge they have not been extended to higher genus.

\paragraph{Towards a discrete version.}
In this paper, we introduce a combinatorial analog of Tutte's barycentric method, in the plane and on surfaces, leading to a purely discrete algorithm.  A first motivation comes from the instability of embeddings under small perturbations.  In the plane, Tutte embeddings can be computed by solving a system of linear equations, but even in that case, the resolution of Tutte embeddings (the maximum distance between two vertices, or between a vertex and a non-incident edge, assuming the minimum edge length is at least one) can be exponential in the number of vertices of the graph~\cite{eg-dspgt-96,df-tfgrp-21}, so rounding may easily create crossings.  The situation is worse on surfaces.  On piecewise-linear surfaces, computing shortest paths exactly requires exact arithmetic to compare sums of square roots, which is not known to be doable in polynomial time~\cite{ehs-ibssr-24}.   Obviously, the situation is even worse on hyperbolic surfaces.  For smooth surfaces of negative curvature, it is not even clear how to represent embeddings exactly.

A second motivation comes from the existing applications of Tutte's theorem to morphing.  In this setting, one cannot expect closed-form descriptions of the morph.  In contrast, in a discrete world, one builds a sequence of local moves from an embedding to the other, in which each intermediate step is an embedding, thus computing a ``discrete isotopy'' between two given embeddings. Our result is a first step towards building such a sequence in which the edge lengths do not vary too much.

More generally, when devising algorithms for topological problems on surfaces, one often faces the following challenge.  On the one hand, these problems are, in principle, easily solved by endowing the surface with a metric of negative curvature, which then enjoys useful properties (this is not possible for sporadic surfaces such as the sphere or the torus, but those can usually be dealt with separately).  On the other hand, such smooth surfaces are not suitable for discrete algorithms.  One thus has to design discrete models that retain the nice properties of smooth surfaces as much as possible.  Such a discretization generally takes the following form: one considers a fixed ``host'' graph (cellularly) embedded on a surface~$S$, and restricts all curves to lie in that graph.  The properties of the host graph, and the definition of a length or of a ``canonical''/``reduced''/shortest path, vary greatly, but a large body of the literature solves topological problems by following this pattern, and we provide a partial review now.

The earliest instance in this vein is probably Dehn's algorithm~\cite{d-tkzf-12} to  determine whether a closed curve on a surface is contractible; this algorithm lies at the root of the entire field of small cancellation theory.  Much more recent works study the same problem and extensions, using finer discrete models such as \emph{octagonal decompositions} by \'E.~Colin de Verdière and Erickson~\cite{CdVE10} to compute shortest homotopic curves, \emph{systems of quads}, introduced by Lazarus and Rivaud~\cite{lazarus2012homotopy} to test homotopy between curves, refined by Erickson and Whittlesey~\cite{ew-tcsr-13} for the same problem, and reused by Despré and Lazarus~\cite{despre2019computing} to compute the geometric intersection number of curves.  Arguably more powerful, and more related to our work, the model of \emph{reducing triangulations}, very recently introduced by \'E.~Colin de Verdière, Despré, and Dubois~\cite{de2024untangling} has been used to decide whether a given drawing of a graph is homotopic to an embedding (and to compute one if it is the case). The host graph is (the 1-skeleton of) a triangulation of a surface (homeomorphic to neither the sphere nor the torus); in this model, curves enjoy the following properties, similar to surfaces of non-positive curvature: each walk  has a unique \emph{reduced} walk homotopic to it, and the property of being reduced is stable by reversal and taking subwalks.  Intuitively, each triangle is equilateral, but there are certain consistent rules to ``break ties'' among equal-length shortest paths.

Under this viewpoint, not every graph that embeds on~$S$ will have an embedding in the host graph, because there may be ``not enough room'' in it.  It is thus natural to consider \emph{weak embeddings}, namely, drawings of graphs in the host graph that can be turned into embeddings on~$S$ by an arbitrarily small perturbation.  In this spirit, Akitaya, Fulek, and Tóth~\cite{akitaya2019recognizing} provide a near-linear time algorithm to decide if a given drawing of a graph is a weak embedding, and if so to produce an actual embedding in a neighborhood of the host graph, encoded combinatorially.

\paragraph{Our contributions.}

Our discrete model of non-positively curved surfaces is that of reducing triangulations~\cite{de2024untangling}.  However, the original model consists of a triangulation with all vertices of degree at least eight, similar to surfaces of negative curvature.  We relax this assumption by allowing degree-six vertices, allowing the discrete analog of ``flat regions'' on surfaces.  The original model is quite constrained; for example, for a surface of fixed genus, there exists only finitely many reducing triangulations homeomorphic to it.  In contrast, our extension allows for finer and finer triangulations of various shapes.  More importantly, we define \emph{harmonious drawings}, which are a natural discrete analog of drawings in which each vertex is in convex position with respect to its neighbors.

Our central result is a purely discrete analog of Tutte's theorem and its generalization to surfaces.  The version for surfaces without boundary is as follows:
\begin{restatable}{theorem}{tuttetheorem}\label{tutte theorem}\RestateRemark
  Let $S$ be an orientable surface without boundary homeomorphic to neither the sphere nor the torus. Let $T$ be a reducing triangulation of $S$. Let $G$ be a graph, and let $f: G \to T^1$ be a harmonious drawing. There is an embedding homotopic to $f$ in $S$ if and only if $f$ is a weak embedding.
\end{restatable}
We emphasize surfaces without boundary as they constitute the hardest cases, but we obtain similar results on all the surfaces with boundary, attaching or not vertices of the drawing to the boundary; this includes the case of the disk, similar to Tutte's original result.  In contrast, the case of the sphere is not relevant in this context (it does not admit reducing triangulations, because it cannot be endowed with a metric of non-positive curvature).  Moreover, our results are not valid on the torus; this case is also very particular, since it admits a flat metric, but it turns out that we need non-positive curvature \emph{and} at least one point of negative curvature; see \autoref{torus problem} below.  Tutte's theorem on the special case of the torus has been studied recently by Erickson and Lin~\cite[Appendix~A and Section~1.1.2]{el-ptmme-23}, who overcome intrinsic difficulties in this case and apply it to morphing; see also Gortler, Gotsman, and Thurston~\cite{ggt-dofma-06}.

Last but not least, we provide an efficient polynomial-time algorithm to make a drawing harmonious.  Ideally, one would like not only to build a single harmonious drawing, but to have some flexibility similar to the choice of the barycentric weights.  Tutte's original method amounts to minimizing the energy of the physical system, namely, the weighted sum of the squares of the lengths of the edges.  Our algorithm actually possesses a stronger property: it  proceeds by ``local'' moves, which never increase the length of any edge of the drawing.  In this way, it allows us to build many harmonious drawings.  In detail:
\begin{restatable}{theorem}{harmonizationtheorem}\label{harmonization theorem}\RestateRemark
  Let $S$ be an orientable surface without boundary homeomorphic to neither the sphere nor the torus. Let $T$ be a reducing triangulation of $S$, with $m$ edges. Let $G$ be a graph, and let $f : G \to T^1$ be a drawing of size $n$. We can compute in $O((m+n)^2n^2)$ time a drawing $f': G \to T^1$, harmonious, homotopic to $f$ in $S$, such that for every edge $e$ of $G$, the image of~$e$ under~$f'$ is not longer than under~$f$.
\end{restatable}

In contrast to the aforementioned works, our input graph~$G$ is arbitrary; we do not need any connectivity requirement, nor require $G$ to be the 1-skeleton of a triangulation.  As a byproduct, we remark that together, these two theorems allow us to decide whether an input drawing of a graph is homotopic to an embedding, in polynomial time, although the running time is not as good as the recent algorithm by É.~Colin de Verdière, Despré, and Dubois~\cite{de2024untangling}: simply turn it into a homotopic harmonious drawing by \autoref{harmonization theorem} and return whether it is a weak embedding using the result of Akitaya, Fulek, and T\'oth~\cite{akitaya2019recognizing}; correctness follows from \autoref{tutte theorem}.

The paper by É.~Colin de Verdière et al.\ is restricted to reducing triangulations with minimum degree at least eight, and does not provide flexibility for computing an embedding if it exists~\cite[Lemma~4.1 and Proposition~4.2]{de2024untangling}.

\paragraph{Overview of the techniques.}

After reviewing some preliminaries in \autoref{sec:preliminaries}, we define (our generalized version of) reducing triangulations, and introduce the notion of harmonious drawings, in \autoref{sec:harmonious}.

\autoref{sec:tutte} is devoted to the proof of our discrete analog of Tutte's theorem (\autoref{tutte theorem}).  In spirit, the proof follows some of the steps of previous proofs of Tutte's theorem~\cite{de1991comment,eh-cti-09}:  we reduce to the case where $f$ is homotopic to an embedding of a triangulation of~$S$.  We then have a continuous map $\varphi$ from a triangulated copy of~$S$ to $S$ itself.  We prove that $\varphi$ is orientation-preserving (or degenerate) on each triangle. We deduce that $f$ can be turned into an embedding $f'$ by homotopy not only in $S$, but even in the neighborhood $\Sigma$ of the 1-skeleton of the reducing triangulation. In~$\Sigma$, we transform $f'$ by isotopy into an embedding arbitrarily close to $f$, thereby proving that $f$ is a weak embedding.
We remark that, since our goal is to prove that $f$ is a weak embedding, we do not have to worry about degenerate cases, which is one of the difficulties in the continuous case.  Rather, we \emph{allow} such degeneracies, which leads to challenges with a different flavor.

In \autoref{sec:harmonizing}, we prove \autoref{harmonization theorem}.  We provide homotopy moves that can be applied to a drawing whenever it is not harmonious, and we describe a procedure that applies those moves in a specific order so that we can prove termination. Substantial technicalities appear since one of these moves does not decrease the length of the drawing strictly, nor, to our knowledge, any kind of potential that would immediately ensure termination of the algorithm.

Finally, in \autoref{boundary}, we extend our results to surfaces with boundary; we reduce to the case of surfaces without boundary by filling the boundary components with surfaces (with genus), thereby extending the reducing triangulation.

\section{Preliminaries}\label{sec:preliminaries}

\subsection{Graphs, drawings, and maps}\label{preliminaries graphs}

In this paper \emphdef{graphs} are undirected, but may have loops and multiple edges. Every graph considered is also topological space. We consider both finite and infinite graphs, all locally finite (each vertex has a finite degree), and we use standard notions of graph theory, such as \emphdef{walks} (finite, semi-infinite, or bi-infinite). Writing the vertices of a bi-infinite walk $W$ as $(\dots w_{-1}, w_0, w_1, \dots)$, we say that $w_0$ is the \emphdef{central vertex} of $W$. The \emphdef{non-negative part} of $W$ is the semi-infinite walk $(w_0, w_1, \dots)$. We denote by $W_0 \cdot W_1$ the concatenation of two walks $W_0$ and $W_1$, and by $W_0^{-1}$ the reversal of $W_0$. Given a closed walk $C$, and $n \geq 0$, we denote by $C^n$ the $n$ times concatenation of $C$ by itself.

A \emphdef{drawing} of a graph $G$ on a graph $H$ is a map $f : G \to H$ that sends every vertex of $G$ to a vertex of $H$, and every edge of $G$ to a walk (possibly a single vertex) in $H$. A drawing $f$ is \emphdef{simplicial} if it sends each edge of $G$ to a vertex or an edge of $H$. A (graph) \emphdef{homomorphism} is a drawing that sends every edge of $G$ to an edge of $H$. Every drawing $f : G \to H$ \emphdef{factors} uniquely as a \emph{simplicial} map $\bar f : \bar G \to H$, where $\bar G$ is a subdivision of $G$, where every edge $e$ of $G$ whose image walk $f \circ e$ has length $n \geq 2$ is subdivided into a path of length $n$ in $\bar G$, and where $e$ is not subdivided otherwise. Also, a simplicial map $f$ \emphdef{factors} uniquely as a \emph{homomorphism} $\hat f : \hat G \to H$ for some graph $\hat G$: the graph~$\hat G$ is obtained from~$G$ by contracting the edges mapped to single vertices, then $f$ corresponds naturally to a homomorphism from $\hat G$ to~$H$. 

Every drawing $f:G\to H$ maps every walk~$W$ in~$G$ to a walk in~$H$, denoted by $f\circ W$.  We also denote by $f(W)$ the image in~$H$ of~$W$ on~$H$.

Given a map $f:A\to B$, and $A'\subseteq A$, $B'\subseteq B$ such that $f(A')\subseteq B'$, we occasionally use the notation $f|_{A'}^{B'}$ to denote the restriction of~$f$ to~$A'$, corestricted to~$B'$, thus yielding a map from $A'$ to~$B'$.

\subsection{Surfaces}

We use standard notions of topology, see any textbook~\cite{armstrong2013basic,s-ctcgt-93} for details.  All the surfaces we consider but the plane are connected, compact, and orientable, so we omit these adjectives in the sequel. Every surface is determined up to homeomorphism by its \emphdef{genus} and number of boundary components. A surface is \emphdef{closed} if its boundary is empty. 

On a surface $S$, a \emphdef{path} is a map $p : [0,1] \to S$; then $p(0)$ and $p(1)$ are the \emphdef{end-points} of $p$. Its \emphdef{reversal} is the map $p^{-1}:[0,1]\to S$ that satisfies $p^{-1}(t)=p(1-t)$. Similarly to walks, we denote by $p_0 \cdot p_1$ the concatenation of two paths $p_0$ and $p_1$. A \emphdef{loop} with \emphdef{basepoint} $b$ is a path whose both endpoints equal $b$. A path is \emphdef{simple} if it is injective, except, of course, that $p(0) = p(1)$ if $p$ is a loop. The \emphdef{relative interior} of a simple path~$p$ is the image of~$(0,1)$ under~$p$. A \emphdef{closed curve} is a map from the circle to $S$; it is simple if it is injective. On a surface $S$ with boundary, an \emphdef{arc} is a path that intersects the boundary precisely at its end-points. Paths and closed curves that differ only by their parameterizations are regarded as equal. 

If $S$ is a surface, when we consider a map $f:G\to S$ we implicitly consider $G$ as a topological space, and $f$ maps the vertices of~$G$ to points on~$S$ and the edges of~$G$ to paths on~$S$.  Moreover, if $W$ is a walk in~$G$, then $f\circ W$ is a curve in~$S$, while $f(W)$ is a subset of~$S$.

A \emphdef{homotopy} between two paths $p_0$ and $p_1$ with the same endpoints is a continuous family
of paths with the same endpoints between $p_0$ and $p_1$. A (free) homotopy between two closed curves $c_0$
and $c_1$ is a continuous family of closed curves between them; this time, no point is required to be fixed.
A loop or closed curve is \emphdef{primitive} if it is not homotopic to another closed curve concatenated $n \geq 2$ times with itself. Given a graph $G$ and a surface $S$, a homotopy between two maps $G \to S$ is a continuous family of maps between them. A homotopy that fixes a subset $X \subset G$ is \emphdef{relative} to $X$. A map $f$ is \emphdef{contractible} if there exists a homotopy of $f$ that turns the image of $f$ into a single point. 

An \emphdef{embedding} of a graph $G$ on a surface $S$ is a ``crossing-free drawing'', a continuous injective map from $G$ to~$S$. The \emphdef{rotation system} of an embedding of $G$ is the cyclic ordering of the edges of $G$ incident to each vertex in the embedding. The \emphdef{faces} of an embedding of $G$ are the connected components of the complement of its image. A \emphdef{triangulation} $T$ is a surface together with a graph embedding in which every face is homeomorphic to an open disk and is bounded by three sides of edges. The \emphdef{1-skeleton} $T^1$ of $T$ is the graph embedded on $T$.

The \emphdef{universal cover} $\widetilde S$ of a closed surface $S$ distinct from the sphere is the plane equipped with a \emphdef{covering map} $\pi : \widetilde S \to S$ that is a local homeomorphism.  A \emphdef{lift} of a point $x \in S$ is a point in the pre-image $\pi^{-1}(x) \subset \widetilde S$. Similarly, a lift of a path $p : [0,1] \to S$ is a path $\widetilde p : [0,1] \to \widetilde S$ such that $\pi \circ \widetilde p = p$.  More generally, we can lift drawings in~$S$ and triangulations in~$S$ in a similar manner.

\subsection{Weak embeddings}\label{sec weak embeddings}

Let $H$ be a graph embedded in a surface $S$, and let $G$ be an abstract graph. A drawing $f : G \to H$ is a \emphdef{weak embedding}~\cite{akitaya2019recognizing} if there are embeddings of $G$ in $S$ that are arbitrarily close to $f$.   As noticed by Akitaya, Fulek, and T\'oth~\cite{akitaya2019recognizing}, the property for $f$ to be a weak embedding can be reformulated, revealing that it does not depend on the actual embedding of $H$ on $S$, but only on $G$, $f$, and the abstract graph $H$ together with its rotation system.  For this reformulation, we sketch the slightly different,  equivalent presentation of \'E.~Colin de Verdière, Despré, and Dubois~\cite[Section~2.2]{de2024untangling}. To ease the reading, we now assume that $H$ has no loop and no multiple edges; the following extends readily to general graphs $H$. The \emphdef{patch system} of $H$ is a surface with boundary $\Sigma$ obtained by “thickening $H$”. (This is similar to the \emph{strip system} introduced by Akitaya et al.~\cite{akitaya2019recognizing}, and to the concept of fat graph~\cite{gt-tgt-87} and ribbon graph~\cite{em-gsdpk-13}; the only difference is that, in our patch systems, each strip corresponding to a ``fat edge'' is contracted to a single path.) Each vertex $v$ of $H$ is assigned a closed disk $D_v$; the disks $D_w$ of the vertices $w$ adjacent to $v$ are attached to $D_v$, along pairwise disjoint segments on the boundary of $D_v$, in the order prescribed by the rotation system of $H$, and respecting the orientations of the disks. Every segment along which two disks $D_v$ and $D_w$ are glued becomes an arc in $\Sigma$ \emphdef{dual to} the edge of $H$ between $v$ and $w$. A drawing $f : G \to H \subset S$ is a \emphdef{weak embedding} if and only if there exists an embedding $\psi : G \to \Sigma$, in general position, that \emphdef{approximates} $f$ in the following sense: for every edge $e$ of $G$, the sequence of arcs of $\Sigma$ crossed by the path $\psi \circ e$ is dual to the sequence of edges of $H$ taken by the walk $f \circ e$.

\subsection{Reducing triangulations}\label{sec:reducing triangulations}

\begin{figure}
    \centering
    \includegraphics[scale=0.5]{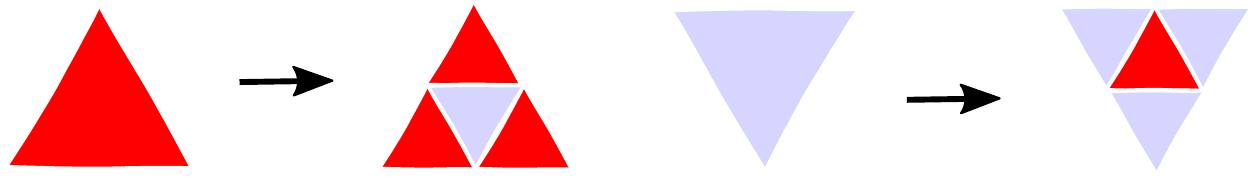}
    \caption{Subdividing the faces of a reducing triangulation.}
    \label{fig:subdivision}
\end{figure}

Our discrete model for surfaces is a slight extension of the notion of \emphdef{reducing triangulation} recently introduced by \'E.~Colin de Verdière, Despré, and Dubois~\cite{de2024untangling}.

We emphasize that the model is restricted to orientable surfaces.  A triangulation $T$ of an (orientable) surface is \emphdef{reducing} if each vertex in the interior of~$T$ has degree at least six, and the dual of~$T$ is bipartite (the faces of~$T$ can be colored red and blue so that adjacent faces receive distinct colors).  Every closed surface other than the sphere admits a reducing triangulation: a reducing triangulation for the torus is made of two triangles, one red and one blue, glued together to form a quadrilateral, whose opposite sides are then identified (\autoref{fig:torus}); for higher-genus surfaces we refer to~\cite[Figure~17]{de2024untangling}. Every surface with boundary admits a reducing triangulation, as can be seen by considering subcomplexes of reducing triangulations of closed surfaces.  (Compared to the original model, our reducing triangulations may have degree-six vertices, which allows in particular to refine an existing reducing triangulation by subdivision; see \autoref{fig:subdivision}.  Also, they may have non-empty boundary.)

Reducing triangulations on surfaces with non-empty boundary will only be considered in \autoref{boundary}, and the following definitions are relevant only for closed surfaces.  A walk $(e_1,e_2)$ of length two in the 1-skeleton~$T^1$ of a reducing triangulation~$T$ makes a \emphdef{turn} at its middle vertex~$v$ (see~\cite[Section~3.1 and Figure~3]{de2024untangling} for details). This turn is a $k$-turn, $k \geq 0$, if $e_2$ results from~$e_1$ by $k$ clockwise rotations of~$e_1$ around~$v$. It is a $-k$-turn if $e_2$ results from $k$ counter-clockwise rotations of $e_1$ around $v$. To be more precise, $(e_1,e_2)$ makes a $k_r$-turn (or $-k_r$-turn) if $e_1$ sees red on its left, and a $k_b$-turn (or $-k_b$-turn) otherwise. The $0$-turns, $1$-turns, $-1$-turns, $2_r$-turns, and $-2_r$-turns are \emphdef{bad}. All other turns are \emphdef{good}.  (We note that these notions depend on an orientation of the surface, which is why we require surfaces to be orientable.) A walk in $T$ is \emphdef{reduced} if it makes no bad turn. Then:
\begin{restatable}{lemma}{reduxwalk}\label{redux walk}\RestateRemark
Let $T$ be a reducing triangulation of the plane. If $x$ and $y$ are vertices of $T$, then there is a unique reduced walk from $x$ to $y$ in $T$.
\end{restatable}

Note that reducing triangulations of the plane cannot have loops or multiple edges (this would otherwise contradict \autoref{redux walk}), but reducing triangulations of more complex surfaces may. \autoref{redux walk} is proved in~\cite[Proposition~3.1]{de2024untangling} on reducing triangulations that do not have degree six vertices. The proof in that paper extends verbatim to our case, only replacing the word “eight” by “six” in the second paragraph of the proof of~\cite[Lemma~3.2]{de2024untangling}. 

\begin{lemma}\label{redux walk 2}
Let $T$ be a reducing triangulation of the plane. Let $C$ be a simple closed walk in $T$, not a single vertex. Orient $C$ so that the bounded side of $C$ lies on its left. Then at least three of the turns of $C$ are a $1$-turn or a $2_r$-turn.
\end{lemma}

\autoref{redux walk 2} can be seen as a consequence of \autoref{redux walk}:

\begin{proof}
By contradiction, assume $C$ has two vertices $x$ and $y$ such that $C$ never makes a $1$-turn or a $2_r$-turn except possibly at $x$ and $y$. Consider the disk $D$ bounded by $C$, and let $T'$ be the reducing triangulation of $D$ that is the restriction of $T$ to $D$. Delete all of $T$ outside $D$. Then extend $T'$ to a new (infinite) reducing triangulation $T''$ of the plane (by adding triangles one by one), in which $C$ does not make any $-1$-turn, nor any $-2_r$-turn. Then the two sub-walks of $C$ between $x$ and $y$ are  distinct reduced walks in $T''$, contradicting \autoref{redux walk}.
\end{proof}

\begin{figure}
    \centering
    \includegraphics[width=0.5\linewidth]{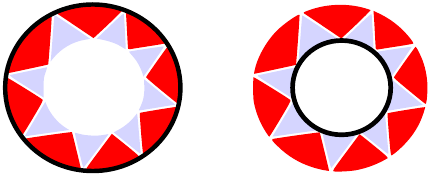}
    \caption{In a reducing triangulation, two homotopic reduced closed walks.  Figure reproduced from~\cite{de2024untangling}.}
    \label{fig:homotopicredux}
\end{figure}

In this paper we say that a \emph{closed} walk in $T$ is reduced if it makes no bad turn. Note that this definition of reduced closed walks departs from that in~\cite{de2024untangling} since, in our setting, a closed walk that makes only $3_r$-turns is reduced. A consequence is that reduced closed walks may not be unique in their free homotopy class: we see the walks depicted in \autoref{fig:homotopicredux} as two distinct freely homotopic reduced closed walks. 

\section{Harmonious drawings}\label{sec:harmonious}

In this section, we provide the key definition of \emph{harmonious drawings} of graphs in a reducing triangulation, used in the statements of our theorems.  We actually start with the notion of strongly harmonious drawings, which are the discrete analog of barycentric drawings, in which each inner vertex is drawn in \emph{convex position}: intuitively, and at a local scale,
every straight line $I$ containing $v$ sees edges incident to $v$ on both sides, here understood in the \emph{strong sense} that some edges incident to $v$ enter the two \emph{open} half-planes separated by $I$ (this is when the surface is homeomorphic to a plane, general surfaces are then handled via their universal cover, as we shall see). Harmonious drawings are a slightly relaxed notion that is suitable for our results.

\begin{cfigure}
    \centering
    \includegraphics[scale=1]{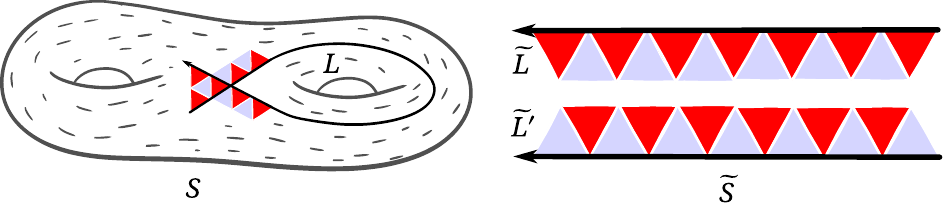}
    \caption{(Left) The closed surface $S$ of genus two, equipped with a reducing triangulation $T$, and a left line $L$ in $T$. (Top Right) The universal covering space $\widetilde S$ of $S$, i.e., the plane, and a left line $\widetilde L$ that lifts $L$ in $\widetilde S$. (Bottom Right) A right line $\widetilde L'$ in $\widetilde S$.}
    \label{fig:line}
\end{cfigure}

\begin{cfigure}
    \centering
    \includegraphics[scale=1]{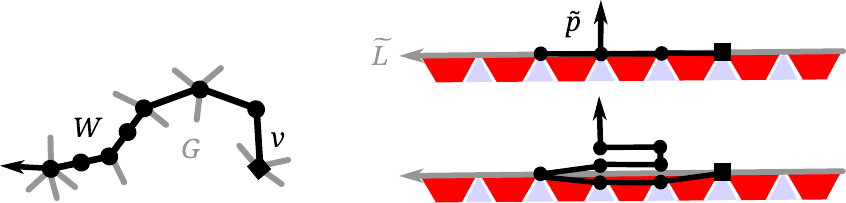}
    \caption{(Left) The vertex $v$ and the walk $W$ in the definition of strong harmony. (Top right) The lifts $\widetilde L$ and $\widetilde p$, when $L$ makes only $3_r$-turns. (Bottom right) The path $\widetilde p$ slightly unpacked to illustrate that $\widetilde p$ can go back and forth, and stagnate.}
    \label{fig:star}
\end{cfigure}

\subsection{Preliminary definitions}

Let $S$ be a closed surface not homeomorphic to the sphere, and let $T$ be a reducing triangulation of $S$.  Let $\widetilde T$ be the (infinite) reducing triangulation that lifts~$T$ in the universal cover~$\widetilde S$ of~$S$.  In $\widetilde T$, we consider a left (resp.\ right) \emphdef{line} to be a bi-infinite walk $\widetilde L$ that makes only $3_r$-turns (resp.\ $-3_r$-turns). See Figures \ref{fig:line}~and~\ref{fig:star}. Note that $\widetilde L$ is reduced, and is thus simple by \autoref{redux walk}. A path $\widetilde p$ starting from the central vertex (\autoref{preliminaries graphs}) of $\widetilde L$ \emphdef{escapes} $\widetilde L$ if $\widetilde p$ enters the right (resp.\ left) side of $\widetilde L$ at some point, and if the prefix of $\widetilde p$ before this point is contained in the non-negative part of $\widetilde L$.

On $S$, the \emphdef{lines} are again the bi-infinite walks (not simple this time, since $T$ is finite) that make only $3_r$-turns or only $-3_r$-turns; note that the lines on $S$ lift to the lines on $\widetilde S$. On $S$, a path $p$ \emphdef{escapes} a line $L$ if there are a lift $\widetilde L$ of $L$, and a lift $\widetilde p$ of $p$ starting from the central vertex of $\widetilde L$, such that $\widetilde p$ escapes $\widetilde L$.

\begin{cfigure}
    \centering
    \includegraphics[scale=1]{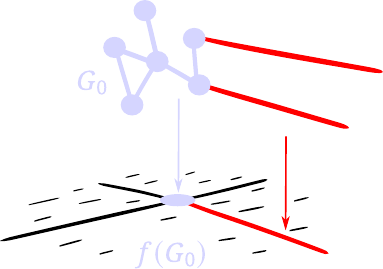}
    \caption{A spur.}
    \label{fig:spurs}
\end{cfigure}

Let $G$ and $M$ be graphs, and let $f : G \to M$ be simplicial.  A \emphdef{cluster} of $f$ is a connected subgraph $G_0$ of $G$ whose edges are all mapped by $f$ to a single vertex of $M$, maximal under this condition. A cluster $G_0$ of $f$ is a \emphdef{spur} if $G_0$ is not a connected component of $G$, and if the directed edges from $G_0$ to $G \setminus G_0$ are all mapped by $f$ to the same directed edge of $M$; see \autoref{fig:spurs}. (Note that the direction of the edges matters when $M$ has loops.)

\subsection{Strongly harmonious drawings}

A simplicial map $f : G \to T^1$ is \emphdef{strongly harmonious} if for every vertex~$v$ of~$G$, and for every line (left or right) $L$ whose central vertex is $f(v)$ in~$T$, there is a walk $W$ based at $v$ in $G$ whose image path $f \circ W$ escapes $L$.  See \autoref{fig:star}. 

Intuitively, $f$ is strongly harmonious if for every vertex $v$ of~$G$, and for every bi-infinite walk~$L$ in~$T^1$ starting at~$f(v)$ making only $3_r$-turns (resp.\ $-3_r$-turns), there is a walk $W$ in~$G$ whose image by $f$ may go ``forward'' on~$L$, sometimes ``backward'' on~$L$, but not to the point it goes before its central point, and then leaves~$L$ to its right (resp.\ left).  But $L$ is non-simple, and periodic since $T$ is finite; $W$ is allowed to wrap, say, 2.3 times forward along~$L$, then 1.8 times backward, and then to leave~$L$ to its right (resp.\ left).  Formalizing this phenomenon seems to be more easily captured using the universal covering space, as above.

\subsection{Harmonious drawings}

\begin{figure}
    \centering
    \includegraphics[scale=1]{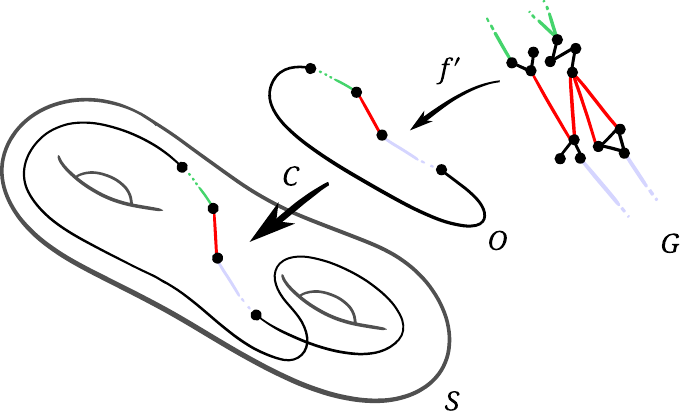}
    \caption{A graph $G$, a cycle graph $O$, and a surface $S$ equipped with a reducing triangulation $T$ (not represented). A simplicial map $f' : G \to O$ without spur, and a reduced closed walk $C : O \to T^1$. The composed drawing $C \circ f'$ is harmonious.}
    \label{fig:harmonicity}
\end{figure}

Harmony is a slightly relaxed version of strong harmony.  On a connected graph $G$, a simplicial map $f : G \to T^1$ is \emphdef{harmonious} if $f$ is strongly harmonious, or if $f = C \circ f'$ for some cycle graph $O$, some reduced closed walk $C : O \to T^1$, and some simplicial map $f' : G \to O$, without spur.  See \autoref{fig:harmonicity}.  In general, a simplicial map $f : G \to T^1$ is harmonious if $f$ is harmonious on every connected component of $G$. 

Finally, recall that every drawing $f:G\to T^1$ factors into a unique simplicial map $\bar f: \bar G \to T^1$. We say that $f$ is harmonious (resp.\ strongly harmonious) if $\bar f$ is.

\subsection{Remarks}

We conclude this section with two remarks. First, we illustrate why strong harmony is not sufficient for our purposes.  The reason is that some drawings of graphs in reducing triangulations cannot be made strongly harmonious by homotopy:  

\begin{restatable}{lemma}{highergenusproblem}\label{higher genus problem}\RestateRemark
Let $S$ be a closed surface, not the sphere nor the torus. There are a reducing triangulation $T$ of $S$, and a closed walk $C$ in $T$, such that every closed walk freely homotopic to $C$ is not strongly harmonious.
\end{restatable}

Second, we give one reason (among others) why our results do not extend to the torus:

\begin{restatable}{lemma}{torusproblem}\label{torus problem}\RestateRemark
There are a reducing triangulation $T$ of the torus, and a closed walk $C$ in $T$, such that every closed walk freely homotopic to $C$ is not reduced.
\end{restatable}

The proofs of Lemmas~\ref{torus problem}~and~\ref{higher genus problem} are deferred to Appendix~\ref{problem proof}.

\section{A Tutte theorem for harmonious drawings: proof of \autoref{tutte theorem}}\label{sec:tutte}

In this section we prove \autoref{tutte theorem}, which we restate for convenience:

\tuttetheorem*

The ``if'' part is trivial, so we focus on the ``only if'' part. If \autoref{tutte theorem} holds on simplicial drawings, then it holds on general drawings by definition. Thus, we assume in this section that drawings are simplicial, with the noticeable exception of some part of \autoref{harmonious and patch systems} where general drawings will be momentarily required: this will be clarified in due course. 

\subsection{Harmonious drawings on patch systems}\label{harmonious and patch systems}

In this section all graphs and drawings are finite, without further mention.  We prove the following, which is roughly a Tutte result for patch systems:

\begin{proposition}\label{tutteinpatch}
Let $M$ be a graph embedded on a surface. Let $G$ be a graph, and let $f : G \to M$ be simplicial  and without spur. If there is an embedding homotopic to $f$ in the patch system of $M$, then $f$ is a weak embedding.
\end{proposition}

The strategy to prove \autoref{tutteinpatch} is to find a sequence of moves (swaps) between the drawing $f$ and a weak embedding, and to prove that applying a move to a weak embedding results in a weak embedding.

In a graph $M$, we say that a (closed) walk $W$ is \emphdef{canonical} if $W$ does not use an edge of $M$ and its reversal consecutively.  We shall use the three following classical facts, see, e.g., Stillwell~\cite[Chapter~2]{s-ctcgt-93}.  (1) If two canonical walks are homotopic in (the topological space associated to)~$M$ relatively to their end-vertices, then they are equal. (2) If two canonical closed walks are freely homotopic, then they differ by a cylic permutation. (3) If two loops $a$ and $b$ based at the same point of $M$ commute, i.e., if there is a loop $c$ such that $a$ is homotopic to $c \cdot b \cdot c^{-1}$ (where the homotopy fixes the basepoint), then $a$ and $b$ are homotopic to powers of a common loop.

\begin{lemma}\label{lem:patch1}
Let $G$ and $M$ be graphs, and let $f : G \to M$ be simplicial. If $G$ is connected, if $f$ is contractible, and if $f$ has no spur, then $f(G)$ is a single vertex of $M$.
\end{lemma}

\begin{proof}
We may assume without loss of generality that $f$ is a homomorphism by contracting the edges that belong to clusters of $f$. Assume by contradiction that $f(G)$ is not a single vertex of $M$. Since $G$ is connected, some edge of $G$ is mapped by $f$ to an edge of $M$. Since $f$ has no spur, there is a semi-infinite walk $W$ in $G$ such that $f \circ W$ is canonical. Since $G$ is finite, there is a subwalk $W'$ of $W$, not a single vertex, that starts and ends at the same vertex of $G$. And the loop $f \circ W'$ is non-contractible in $M$ by (1), contradicting the assumption that $f$ is contractible.
\end{proof}

Until the end of this section, we need to consider general drawings (instead of simplicial ones). Recall from \autoref{sec:preliminaries} that every drawing $f : G \to M$ factors as a simplicial map $\bar f : \bar G \to M$, where $\bar G$ is a subdivision of $G$. Observe that every cluster of $f$ corresponds to a cluster in $\bar f$. In addition, the vertices of $\bar G$ inserted in the edges of $G$ are clusters of $\bar f$. There is no other cluster in $\bar f$.

\begin{lemma}\label{lem:patch2}
Let $M$ be a graph embedded on a surface. Let $G$ be a graph, and let $f : G \to M$ be a drawing. If there is an embedding homotopic to $f$ in the patch system of $M$, then there is a weak embedding $f' : G \to M$, homotopic to $f$, that has no spur.
\end{lemma}

\begin{figure}
    \centering
    \includegraphics[scale=1]{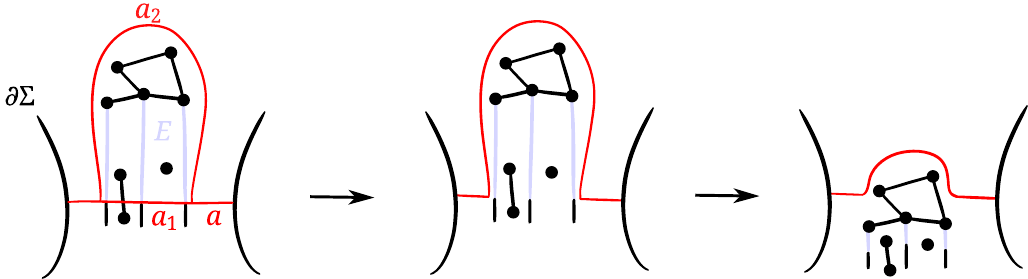}
    \caption{In the proof of \autoref{lem:patch2}, if $f'$ had a spur, then the number of crossings of $g$ with the arcs of~$\Sigma$ could be decreased.}
    \label{fig:no spur}
\end{figure}

\begin{proof}
Let $\Sigma$ be the patch system of $M$.  We regard $f$ as a map from $G$ to~$\Sigma$.  There is by assumption an embedding $g:  G \to \Sigma$ homotopic to $f$. We can assume that $g$ crosses the arcs of~$\Sigma$ as few times as possible subject to the constraint that it is an embedding homotopic to~$f$.  Let $f' : G \to M$ be the drawing of which $g$ is an approximation. Then $f'$ is a weak embedding homotopic to $f$.  

There remains to prove that $f'$ has no spur. By contradiction, assume that it does. See \autoref{fig:no spur}. Let $G_0$ be the cluster of the spur, and let $a$ be the arc of $\Sigma$ dual to the directed edge of the spur. Let $E \subset \Sigma$ contain, for every edge $e$ directed from $G_0$ to $G \setminus G_0$, the prefix of the image path $g(e)$ that leaves $g(G_0)$ to reach its first crossing with $a$. Let $a_1$ be the subpath of~$a$ that starts just before its first crossing with~$E$ and ends just after its last crossing with~$E$. 
Because $g$ is an embedding, there exists a simple path~$a_2$ in $\Sigma$ with the same endpoints as~$a_1$, otherwise disjoint from the arcs of $\Sigma$ and from $g(G)$, such that the disk bounded by $a_1$ and $a_2$ contains $g(G_0)$. We consider a self-homeomorphism~$h$ of~$\Sigma$ that affects only a neighborhood of the disk bounded by~$a_1$ and~$a_2$, and pushes $a_2$ to~$a_1$.  Then $h\circ g$ is an embedding of~$G$ homotopic to~$g$, with fewer crossings with the arcs of~$\Sigma$, a contradiction.
\end{proof} 

\begin{figure}
    \centering
    \includegraphics[scale=1]{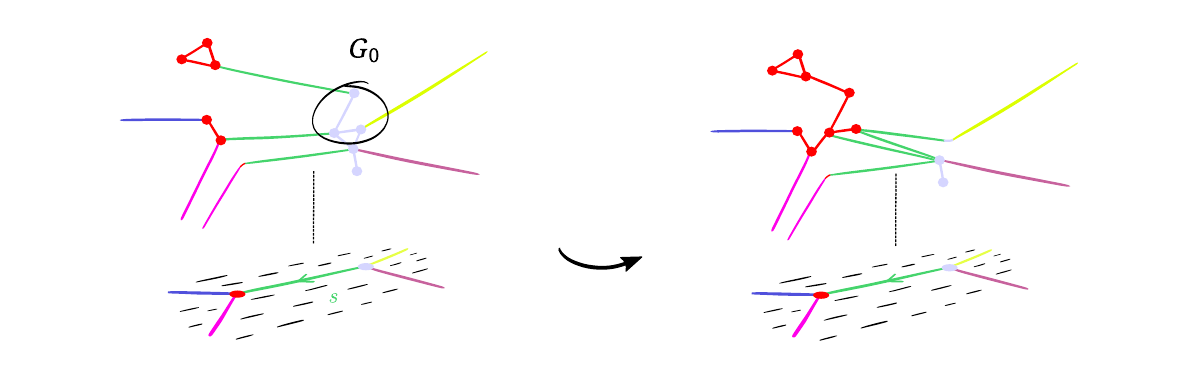}
    \caption{Swap of the subgraph $G_0$ along the directed edge $s$.}
    \label{fig:swaps}
\end{figure}

A drawing $f : G \to M$ is \emphdef{essential} if there is no connected component of $G$ on which $f$ is contractible. Let $f : G \to M$ be an essential drawing. Let $G_0$ be an induced subgraph of $G$, connected and mapped to a single vertex of $M$ by $f$. Let $s$ be a directed edge of $T$ based at the vertex $f(G_0)$. Consider the following operation that transforms $f$ homotopically into another essential drawing $f' : G \to M$ (\autoref{fig:swaps}). First slide $f(G_0)$ along $s$. Then, among the image walks of the edges directed from $G_0$ to $G \setminus G_0$, those that initially admitted $s$ as a prefix now start with the concatenation of $s^{-1}$ and $s$: shorten those walks by removing this prefix. In the particular case where $f$ and $f'$ both have no spur (in addition of being essential), we call this operation a \emphdef{swap}. A key but trivial observation is that, if $f'$ results from a swap of $f$, then $f$ results from a swap of $f'$.

Note that if $f$ has no spur, if $G_0$ is a cluster of $f$, and if $s$ is the image by $f$ of some edge directed from $G_0$ to $G \setminus G_0$, then $f'$ has no spur, so the operation is indeed a swap in this case. Every swap constructed in this section will either be such a swap, or the “inverse” of such a swap.

\begin{lemma}\label{lem:patch3}
Let $G$ and $M$ be graphs, and let $f,f' : G \to M$ be drawings. If $f$ and $f'$ are essential and without spur, and if they are homotopic, then there is a sequence of swaps between $f$ and $f'$. 
\end{lemma}

\begin{proof}
We shall perform swaps on $f$ and $f'$ so that in the end $f = f'$. Assume without loss of generality that $G$ is connected, and fix a spanning tree $Y$ of $G$. We claim that there is a sequence of swaps that modify $f$ so that in the end $f(Y)$ is a single vertex of $M$. To prove the claim, we say that an edge $e$ of $Y$ is contracted if $f(e)$ is a single vertex of $M$. We choose an arbitrary root for~$Y$, and given any edge $e$ of $Y$, we refer to the edges of $Y$ distinct from $e$ and separated from the root by $e$ as the descendants of $e$. We prove the claim by considering an edge $e$ of $Y$, by assuming that $e$ is not contracted and that every descendant of $e$ is contracted, and by exhibiting a swap that shortens $f \circ e$ while keeping the descendants of $e$ contracted. Direct $e$ so that the tail vertex of $e$ is the closest to the root of $Y$. Let $G_0$ be the cluster containing the head vertex of $e$. Since the descendants of $e$ are all contracted, they all belong to $G_0$. Since $e$ is not contracted, we may consider the last directed edge $s$ of the walk $f \circ e$. The swap of $f(G_0)$ along the reversal of $s$ shortens $f \circ e$, and keeps the descendants of $e$ contracted.  This proves the claim.

We use the claim immediately on both $f$ and $f'$, then contract the spanning tree $Y$ in both drawings. Every edge $e$ of $G$ not in $Y$ becomes a loop. The image loop $f \circ e$ is contractible if and only if $f' \circ e$ is contractible, and in that case $f(e) = f(Y)$ and $f'(e) = f'(Y)$ are single vertices of $M$ by (1). Also, given any two edges $e_1$ and $e_2$ of $G$ not in $Y$, we have $f \circ e_1 \simeq f \circ e_2$ if and only if $f' \circ e_1 \simeq f' \circ e_2$, where $\simeq$ denotes the homotopy of loops relatively to their basepoint. In that case $f \circ e_1 = f \circ e_2$ and $f' \circ e_1 = f' \circ e_2$ by (1). By contracting the contractible loops and identifying the homotopic loops in both $f$ and $f'$, we may assume that every connected component of~$G$ has a single vertex (all its edges are loops), and that each of $f$ and $f'$ maps the edges of $G$ to pairwise distinct non-trivial walks in $M$.

The end of the proof is adapted from the proof of~\cite[Lemma~8.6]{de2024untangling}, to which we refer for details. Let $v$ be a vertex of $G$. Since $f$ and $f'$ are essential, the graph~$G$ is not a single vertex, so $G$ has a loop $e$.  At this point, $f\circ e$ and $f'\circ e$ are canonical walks, but not necessarily canonical closed walks; however, by performing swaps on $f$ and $f'$ at $v$ a few times if needed, we enforce that $f \circ e$ and $f' \circ e$ are canonical closed walks. Then by (2), and since $f \circ e$ and $f' \circ e$ are freely homotopic, $f \circ e$ is a cyclic permutation of $f' \circ e$.  By performing swaps on $f$ again a few times, we enforce that $f \circ e$ and $f' \circ e$ are actually equal (not up to cyclic permutation). At this point, we claim that we can slide  $f(v)$ around $f \circ e$ by performing swaps on $f$, so that in the end $f$ and $f'$ are homotopic relatively to $v$. This claim implies $f = f'$ by (1), which proves the lemma.  There remains to prove the claim.  For this, consider the path followed by $f(v)$ during some \emph{free} homotopy from $f$ to $f'$. This path is a loop based at $f(v)$.  Let $W$ be the canonical walk homotopic to it. We prove the claim by showing that $f \circ e$ and $W$ are equal to powers of the same walk. Indeed, let $C$ be a primitive canonical closed walk such that $f \circ e$ is freely homotopic to a power of $C$. Without loss of generality $f \circ e$ is equal to a power of $C$ by (2). Also, $W$ commutes with $f \circ e$ since $f \circ e = f' \circ e$, and since $f \circ e$ is homotopic to the walk $W \cdot (f' \circ e) \cdot W^ {-1}$. Thus, the walk $W$ is homotopic to a power of $C$ by (3). And so $W$ is equal to this power of $C$ by (1). 
\end{proof}

\begin{lemma}\label{lem:patch4}
Let $M$ be a graph embedded on a surface. Let $G$ be a graph, and let $f,f' : G \to M$ be drawings. If $f$ and $f'$ are essential, if $f'$ results from a swap of $f$, and if $f$ is a weak embedding, then $f'$ is a weak embedding.
\end{lemma}

\begin{proof}
Let $\Sigma$ be the patch system of $M$, and let $F : G \to \Sigma$ be an embedding that approximates $f$ in $\Sigma$. Let $G_0$ be the subgraph of the swap, and let $a$ be the arc of $\Sigma$ dual to the directed edge of the swap. Let $E \subset \Sigma$ contain, for every edge $e$ directed from $G_0$ to $G \setminus G_0$, the prefix of the image path $F(e)$ that leaves $F(G_0)$ to reach either a vertex of $F(G \setminus G_0)$ in the same face of $\Sigma$, or its first crossing with an arc of $\Sigma$. Let $E_1 \subset E$ contain the paths that reach $a$. Then $E_1 \neq \emptyset$, for otherwise $G_0$ would either be a spur in $f'$, or it would be a cluster mapped to a single vertex in $f'$. Let $a_1$ be the subpath of~$a$ that starts just before its first crossing with~$E_1$ and ends just after its last crossing with~$E_1$. Because $F$ is an embedding, there exists a simple path~$a_2$ with the same endpoints as~$a_1$, otherwise disjoint from the arcs of $\Sigma$, and disjoint from $F(G)$ except for each path in $E \setminus E_1$ that $a_2$ may cross at most once, such that the disk $D$ bounded by $a_1$ and $a_2$ contains $F(G_0)$.

We claim that $D$ does not contain any part of $F(G)$ other than $F(G_0)$ and $E$. By contradiction assume that it does. Then $D$ contains the image $F(v)$ of a vertex $v$ of $G \setminus G_0$, since $G_0$ is an induced subgraph of $G$, and since any edge of $F(G \setminus G_0)$ intersecting $D$ must have an endpoint in $D$. If a path based at $F(v)$ in $F(G)$ does not intersect $a_1$ nor $F(G_0)$, then this path stays in the interior of $D$, and so it does not intersect any arc of $\Sigma$. Therefore, in $f'$, the cluster containing $v$ is either mapped to a single vertex, or is a spur, which is a contradiction.

Now consider a self-homeomorphism~$h$ of~$\Sigma$ that affects only a neighborhood of $D$, and pushes $a_2$ to~$a_1$ (similar to the one depicted in \autoref{fig:no spur}). Then $h\circ F$ is an embedding of~$G$, that approximates $f'$ by the preceding claim.
\end{proof}

The following lemma is easy and might be folklore, but we could not find a reference, so we provide a proof for completeness:

\begin{lemma}\label{lem:planar part}
Let $S$ be a surface. Let $G$ be a graph, and let $f : G \to S$ be a map. If $f$ is contractible, and if there is an embedding homotopic to $f$, then $G$ is planar.
\end{lemma}

\begin{proof}
Without loss of generality $G$ is connected. Fix a vertex $r$ and a spanning tree $T$ of $G$. There is an embedding $f' : G \to S$ homotopic to $f$. In $f'$, almost contract the image of $T$ by isotopy, without changing the image of $r$, in order to push the image of $T$ inside a small neighborhood $N$ of $f(r)$. Every edge $e$ of $G$ not in $T$ is mapped by $f'$ to a simple contractible loop. Those loops can be pushed inside $N$ by isotopy since, by a result of Epstein~\cite[Theorem~1.7]{e-c2mi-66}, each of them bounds a disk with only (possibly) contractible loops inside it.
\end{proof}

Finally, we prove \autoref{tutteinpatch}:

\begin{proof}[Proof of Proposition~\ref{tutteinpatch}]
  If $G'$ is a connected component of $G$ on which $f$ is contractible, then $f(G')$ is a single vertex of $M$ by \autoref{lem:patch1}, and $f|_{G'}$ can made an embedding in an arbitrarily small disk in the patch system of $\Sigma$ by~\autoref{lem:planar part}. So we can assume that $f$ is essential. By \autoref{lem:patch2}, there is a weak embedding $g:G\to M$, homotopic to~$f$, without spur. By \autoref{lem:patch3}, there is a sequence of swaps from $g$ to~$f$.  By \autoref{lem:patch4}, all the maps from $G$ to~$M$ in this sequence are weak embeddings, thus $f$ is itself a weak embedding, as desired.
\end{proof}

\subsection{A property of the coherently oriented maps homotopic to the identity}

Let $S$ be a surface, and let $Y \subset S$ be finite. A map $\varphi : S \to S$ is \emphdef{coherently oriented} at $Y$ if $\varphi ^{-1}(Y)$ is finite and if $\varphi$ is, locally, an orientation-preserving homeomorphism around every point of~$\varphi^{-1}(Y)$, or an orientation-reversing homeomorphism around every point of~$\varphi^{-1}(Y)$.

The following proposition is proved using the topological notion of the degree of a self-map $\varphi:S\to S$.  It will be used in the case where $Y$ contains one point per face of $T$, and we will regard $S \setminus Y$ as the patch system of $T^1$. The cardinality of a set $Z$ is denoted by $\# Z$.

\begin{restatable}{proposition}{propselfmaps}\label{prop:selfmaps}
Let $S$ be a closed surface. Let $\varphi  : S \to S$ be a map, homotopic to the identity map of $S$. Let $Y \subset S$ be finite. If $\varphi$ is coherently oriented at $Y$, then $\# \varphi^{-1}(Y) = \# Y$ and $\varphi|_{S \setminus \varphi^{-1}(Y)}^{S \setminus Y}$ is homotopic to a homeomorphism $S \setminus \varphi^{-1}(Y) \to S \setminus Y$.
\end{restatable}

\begin{proof}
Let $y \in Y$. We claim that $\varphi^{-1}(y)$ contains only one point, and that $\varphi$ is orientation-preserving around this point. To prove this claim, let $n^+$ and $n^-$ denote the number of points of $\varphi^{-1}(y)$ around which $\varphi$ is respectively orientation-preserving and orientation-reversing. The difference $n^+ - n^-$ does not depend on the choice of $y$ (as long as $y$ is chosen so that $\varphi$ is locally a homeomorphism around every point of $\varphi^{-1}(y)$) and is known as the degree of $\varphi$.

The degree of a map is invariant by homotopy, and $\varphi$ is homotopic to the identity of $S$, so $n^+ - n^- = 1$. We assumed $n^- = 0$ or $n^+=0$. Thus, $n^- = 0$ and $n^+ = 1$.

Using our claim, for every $y \in Y$, there is a closed disk $B_y \subset S$ containing $y$ in its interior, such that $\varphi^{-1}(B_y)$ is a closed disk $A_y$, and such that $\varphi|_{A_y}^{B_y}$ is an orientation-preserving homeomorphism. Without loss of generality, the disks $\{B_y\}_{y \in Y}$ are pairwise disjoint. Let $N$ be the surface obtained from $S$ by removing the interiors of the disks $\{B_y\}_{y \in Y}$, and $M$ be obtained by removing the interiors of the disks $\{A_y\}_{y \in Y}$. The map $\varphi' := \varphi|_M^N$ is defined. Since $\varphi$ is a degree one map, $\varphi' : M  \to N$ is a degree one map. By construction, $\varphi'$ maps the boundary of $M$ to the boundary of $N$ and the interior of $M$ to the interior of $N$, and the restriction and corestriction of $\varphi'$ to the boundaries of $M$ and $N$ is a homeomorphism. It follows from a result by Edmonds~\cite[Theorem~4.1]{edmonds1979deformation} that $\varphi'$ is homotopic to a homeomorphism $M \to N$, where the homotopy is relative to~$\partial M$. (More precisely, this result follows from \cite[Theorem~3.1]{edmonds1979deformation} by noting that (1) each branch covering of degree $\pm1$ is a homeomorphism, and that (2) each pinch map from $M$ to~$M$ is homotopic to the identity, because the simple closed curve defining the pinch must bound a disk.)
\end{proof}

\subsection{From graphs to triangulations}

Let $T$ be a reducing triangulation, and let $Z$ be a triangulation. A map $\varphi : Z \to T$ is \emphdef{simplicial} if $\varphi$ maps $Z^1$ to $T^1$ simplicially, and if $\varphi$ sends every face of $Z$ to a vertex, an edge, or a face of $T$. It is \emphdef{strongly harmonious} if $\varphi|_{Z^1}^{T^1}$ is strongly harmonious.

In this section, we prove the following proposition, which (essentially) allows us to consider simplicial drawings of entire triangulations, instead of graphs.

\begin{restatable}{proposition}{propextension}\label{prop:extension}
Let $S$ be a closed surface distinct from the sphere. Let $T$ be a reducing triangulation of $S$. Let $G$ be a finite graph embedded in $S$, and let $f: G \to T^1$ be simplicial. Assume that $f$ is strongly harmonious, and that $f$ is homotopic to the inclusion map $G \hookrightarrow S$. There are a triangulation $Z$ of $S$ whose 1-skeleton contains a subdivision of $G$ as a subgraph, and a simplicial map $\varphi : Z \to T$ with $\varphi|_G^{T^1} = f$, such that $\varphi$ is strongly harmonious and homotopic to the identity map of $S$.
\end{restatable}

The proof of \autoref{prop:extension} relies on two lemmas:

\begin{lemma}\label{lem:tight extension}
Let $S$ be a closed surface distinct from the sphere. Let $T$ be a reducing triangulation of $S$. Let $G$ be a graph obtained from another graph $G'$ by inserting a path graph $P$ between (possibly equal) vertices of $G'$. Let $f : G \to T^1$ be simplicial. If $f|_{G'}$ is strongly harmonious, and if $f$ maps $P$ to a reduced walk in $T$, then $f$ is strongly harmonious.
\end{lemma}

\begin{proof}
If $P$ is a single edge, then it does not affect strong harmony. So assume that $P$ has an interior vertex $v$, and let $L$ be a left line (the right line case being similar) in $T$ whose central vertex is $f(v)$. We will exhibit a walk based at $v$ in $G$ whose image path escapes $L$. This will prove the lemma. Consider the universal covering triangulation $\widetilde T$ of $T$, and a lift $\widetilde L$ of $L$ in $\widetilde T^1$. Consider also the reduced walk $X := f|_P$. There are a lift $\widetilde P$ of $P$, a lift $\widetilde v$ of $v$ in $\widetilde P$, and a lift $\widetilde X : \widetilde P \to \widetilde T^1$ of $X$, such that $\widetilde X(\widetilde v)$ is the central vertex of $\widetilde L$. Let $\widetilde w_0$ and $\widetilde w_1$ be the two end-vertices of $\widetilde P$, and let $\widetilde P_0$ and $\widetilde P_1$ be the two sub-walks of $\widetilde P$ that go from $\widetilde v$ to respectively $\widetilde w_0$ and $\widetilde w_1$. Since $\widetilde X$ is a reduced walk, and since $\widetilde L$ makes only $3_r$-turns, one of $\widetilde P_0$ and $\widetilde P_1$, say $\widetilde P_0$ without loss of generality, is such that $\widetilde X \circ \widetilde P_0$ either escapes $\widetilde L$, or stays in the non-negative part of $\widetilde L$. In the first case, considering the vertex $w_0$ of $G$ lifted by $\widetilde w_0$, the portion $P_0$ of $P$ from $v$ to $w_0$ is such that $f \circ P_0$ escapes $L$. In the latter case, replace the central vertex of $\widetilde L$ by $\widetilde X(\widetilde w_0)$, and project the resulting line onto the surface, obtaining a line $L'$ which is a shift of $L$. Since $w_0$ belongs to $G'$, and since $f|_{G'}$ is strongly harmonious, there is a walk $Q$ based at $w_0$ in $G$ such that $f \circ Q$ escapes $L'$. Then the concatenation $P_0'$ of $P_0$ and $Q$ is such that $f \circ P_0'$ escapes $L$.
\end{proof}

\begin{lemma}\label{lem:flathyperbolic}
Let $w,v_0,v_1$ be pairwise distinct vertices of a plane reducing triangulation $T$. If $v_0$ and $v_1$ are adjacent in $T$, then there is $i \in \{0,1\}$ such that, along the reduced walk from $v_i$ to $w$, the vertex consecutive to $v_i$ is adjacent to or equal to $v_{1-i}$.
\end{lemma}
\begin{proof}
  Let $e$ be the edge of $T$ between $v_0$ and $v_1$; direct~$e$ so that it sees blue at its left, and assume without loss of generality that it is directed from~$v_0$ to~$v_1$. Let $W_1$ be the reduced walk from $v_1$ to $w$. Assume that the vertex consecutive to $v_1$ in $W_1$ is neither equal nor adjacent to $v_0$, for otherwise there is nothing to do. Let $W_0$ be the concatenation of $e$ with $W_1$.  By assumption, $W_0$ does not make a $0$-turn, $1$-turn, or $-1$-turn at~$v_1$, and by our choice of direction of~$e$, it also does not make a $2_r$-turn, so it is reduced. And $W_0$ starts with edge~$e$, as desired.
\end{proof}

\begin{figure}
    \centering
    \includegraphics[scale=0.3]{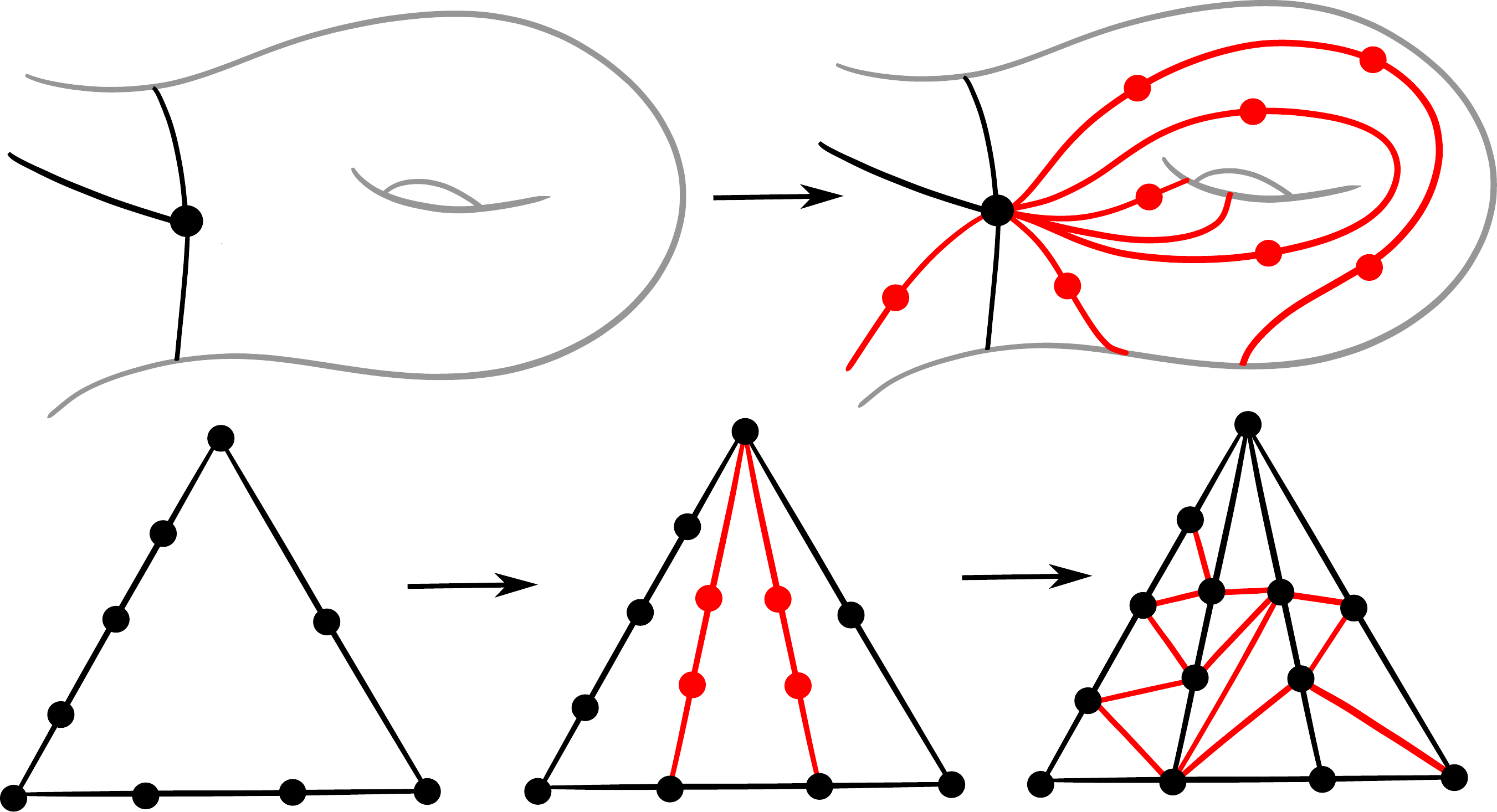}
    \caption{In the proof of \autoref{prop:extension}, the embedded graph $G$ is extended to a triangulation.}
    \label{fig:refinement}
\end{figure}

\begin{proof}[Proof of Proposition~\ref{prop:extension}]
  In this proof, we will use the following standard topological fact: Every graph embedded on a surface can be extended to a triangulation by adding edges (not vertices). This can be done by repeatedly inserting an edge~$e$ inside a face~$F$, where $e$ “differs” from every edge $e'$ on the boundary of $F$ in the sense that the concatenation of $e$ and $e'$ does not bound a disk in $F$.

  At any time we denote by $G^\lozenge$ the graph derived from $G$ by removing its degree two vertices.  The proof is in three steps; see \autoref{fig:refinement}.  In the first step, as long as some face of $G^\lozenge$ is not a triangle, we insert an edge $e$ in this face, as described in the preceding paragraph. Doing so, we consider the homotopy from the inclusion map $G \hookrightarrow S$ to the map $f$, and apply this homotopy to the two end-vertices of $e$, thus extending $e$ to a path $e'$ between vertices of $T$. Let $p$ be the unique reduced path homotopic to $e'$ (\autoref{redux walk}). If the length $n$ of $p$ is greater than one, then we insert $n-1$ vertices in $G$ along $e$. We also extend $f$ to $e$ by mapping $e$ to $p$. In this way, $f$ remains homotopic to the inclusion map $G \hookrightarrow S$, and $f$ remains strongly harmonious by \autoref{lem:tight extension}.  Now, every face of $G^\lozenge$ is a triangle.

  In the second step (see \autoref{fig:refinement}, bottom center), for every (triangular) face~$m$ of $G^\lozenge$, we consider an edge $e$ incident to $m$ in $G^\lozenge$. If $e$ is subdivided in $G$, then we insert a path in $m$ from each interior vertex of $e$ to the vertex of $m$ opposite to $e$.  As in the previous paragraph, we map each of the new paths to a reduced path in~$T$, so that $f$ remains homotopic to the inclusion map $G \hookrightarrow S$, and so that $f$ remains strongly harmonious. Now, every (triangular) face~$m$ of $G^\lozenge$ is incident to an edge that is not subdivided in $G$.
  
  In the last step, \autoref{lem:flathyperbolic} ensures that we can triangulate $m$ by inserting new edges in $m$ (see \autoref{fig:refinement}, bottom right), and by mapping in $f$ each such new edge to \emph{a vertex or an edge} of $T$, keeping $f$ strongly harmonious and homotopic to the inclusion map $G \hookrightarrow S$. Now $G$ is the 1-skeleton of a triangulation $Z$, and $f$ trivially extends to a simplicial map $\varphi : Z \to T$.
\end{proof}

\subsection{A property of the maps homotopic to the identity and strongly harmonious}

Here is another key step towards the proof of \autoref{tutte theorem}: we prove that our simplicial drawings of triangulations orient the (non-degenerate) triangles coherently.

\begin{restatable}{proposition}{proptighttriang}\label{prop:tighttriang}
Let $S$ be a closed surface distinct from the sphere. Let $T$ be a reducing triangulation of $S$. Let $Z$ be a triangulation of $S$, and let $\varphi : Z \to T$ be simplicial. If $\varphi$ is strongly harmonious, and if $\varphi$ is homotopic to the identity map of $S$, then there cannot be two faces $z_+$ and $z_-$ of $Z$ for which $\varphi|_{z_+}$ is positive and $\varphi|_{z_-}$ is negative.
\end{restatable}

In this section only, it is convenient to consider the plane $P$, and an (infinite) reducing triangulation $T$ of $P$; the notion of strong harmony immediately extends to that setting. And again, if $Z$ is a triangulation of $P$, a simplicial map $\varphi : Z \to T$ is strongly harmonious if $\varphi|_{Z^1}^{T^1}$ is strongly harmonious. 

In the 1-skeleton $T^1$ of $T$, the minimum number of edges of a path between two given vertices defines a distance on the vertex set of $T$. A map $\varphi : P \to P$ is \emphdef{uniformly homotopic} to the identity map $1_{P \to P}$ if it is homotopic to the identity and there is $\kappa > 0$ such that every point $x \in P$ mapped to a vertex $v$ of $T$ by $\varphi$ lies in a vertex, edge, or face of $T$ whose incident vertices are at distance less than $\kappa$ from $v$. We shall prove the following:

\begin{restatable}{proposition}{proptutteinplane}\label{prop:tutteinplane}
Let $P$ be the plane. Let $T$ be a reducing triangulation of $P$. Let $Z$ be a triangulation of $P$, and let $\varphi : Z \to T$ be simplicial. If $\varphi$ is strongly harmonious, and if $\varphi$ is uniformly homotopic to the identity map of $P$, then there cannot be two faces $z_+$ and $z_-$ of $Z$ for which $\varphi|_{z_+}$ is positive and $\varphi|_{z_-}$ is negative.
\end{restatable}

\autoref{prop:tighttriang} (concerning surfaces) easily follows from \autoref{prop:tutteinplane} (concerning the plane) by lifting:

\begin{proof}[Proof of Proposition~\ref{prop:tighttriang}, assuming Proposition~\ref{prop:tutteinplane}]
The universal covering space $\widetilde S$ of $S$ is the plane. Also, $T$ lifts to a reducing triangulation $\widetilde T$ of $\widetilde S$, $Z$ lifts to a triangulation $\widetilde Z$ of $\widetilde S$, and $\varphi$ lifts to a simplicial map $\widetilde \varphi : \widetilde Z \to \widetilde T$. Since $\varphi$ is strongly harmonious, $\widetilde \varphi$ is strongly harmonious. We claim that $\widetilde \varphi$ is uniformly homotopic to the identity map of $\widetilde S$. This claim implies by \autoref{prop:tutteinplane} that there cannot be two faces $z_+$ and $z_-$ of $\widetilde Z$ for which $\widetilde \varphi|_{z_+}$ is positive and $\widetilde \varphi|_{z_-}$ is negative. Then the same holds for $Z$ and $\varphi$, proving the proposition.

To prove the claim, we use the assumption that there is a homotopy $H$ between $\varphi$ and the identity map of $S$, and we lift $H$ to a homotopy $\widetilde H$ between $\widetilde \varphi$ and the identity map of $\widetilde S$. For every lift $\widetilde x \in \widetilde S$ of a point $x \in S$, the path $\widetilde p$ followed by the image of $\widetilde x$ in $\widetilde H$, and the path $p$ followed by the image of $x$ in $H$ are such that the number of vertices, edges, and faces of $\widetilde T$ crossed by $\widetilde p$ is equal to the number of vertices, edges, and faces of $T$ crossed by $p$. The latter is uniformly bounded since $S$ is compact.
\end{proof}

\begin{figure}
    \centering
    \includegraphics[scale=1]{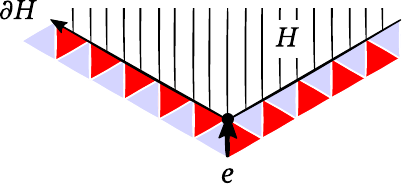}
    \caption{A cone.}
    \label{fig:cone}
\end{figure}

The rest of this section is devoted to the proof of \autoref{prop:tutteinplane}. We need some definitions and lemmas. See \autoref{fig:cone}. Let $T$ be a reducing triangulation of the plane. The part of $T$ on a given side of a bi-infinite reduced path $\partial H$, and not on $\partial H$, is a \emphdef{half-plane} $H$ of $T$, and $\partial H$ is the boundary of $H$. We emphasize that half-planes are open. A half-plane $H$ is \emphdef{nested} in another half-plane $H'$ if $H \subset H'$ and $\partial H \cap \partial H' = \emptyset$. Let $e$ be a directed edge $e$ of $T$ that sees blue on its left, and let $L$ be the bi-infinite reduced walk in $T$ that contains the head vertex $v$ of $e$, and makes only $3_b$- and $3_r$-turns, except at $v$ where it makes a $4_b$-turn whose middle edge is $e$. The \emphdef{cone} of $e$ is the half-plane $H$ on the right of $L$, and $v$ is the \emphdef{tip} of $H$.

\begin{lemma}\label{lempush}
Let $T$ be a reducing triangulation of the plane. If $H$ is a half-plane of $T$, and if $v \in H$ is a vertex of $T$, then $v$ is the tip of a cone $H'$ nested in $H$.
\end{lemma}

\begin{cfigure}
    \centering
    \includegraphics[scale=1]{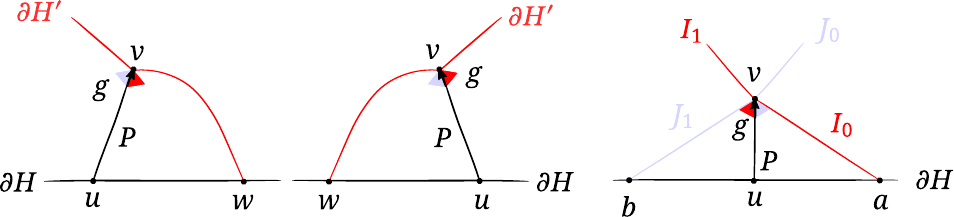}
    \caption{The impossible cases in the proof of \autoref{lempush}.}
    \label{fig:use cones}
\end{cfigure}

\begin{proof}
There is a reduced path $P$ between $v$ and a vertex $u$ of $\partial H$, such that $P$ is internally included in $H$. Let $g$ be the edge of $P$ incident to $v$, directed toward $v$. First assume that $g$ sees blue on its left. We claim that the cone $H'$ of $g$ is nested in $H$. First we prove $\partial H' \cap \partial H = \emptyset$ by contradiction. See \autoref{fig:use cones} Assuming that $\partial H$ and $\partial H'$ share a vertex $w$, let $Q'$ be the concatenation of $P$ and of the subpath of $\partial H'$ between $v$ and $w$, and let $Q$ be the subpath of $\partial H$ between $u$ and $w$. Then $Q$ and $Q'$ are distinct reduced paths with the same end-vertices in $T$, contradicting \autoref{redux walk}.  It then easily follows that $H' \subset H$, and we now provide the details. We have $\partial H' \subset H$ since $\partial H' \cap H \neq \emptyset$ and $\partial H' \cap \partial H = \emptyset$. The reduced path $P$ is internally disjoint from $\partial H$ and $\partial H'$ by \autoref{redux walk}, and since $\partial H$ and $\partial H'$ are reduced. So the relative interior of $P$ lies in the open region $R$ between $\partial H$ and $\partial H'$. Also, the relative interior of $g$ is disjoint from $H'$, and so is $R$.

Now assume that $g$ sees red on its left. We claim that there is a rotation of one turn of $g$ around its head vertex $v$ (either clockwise or counter-clockwise) after which the cone $H'$ of $g$ satisfies $\partial H' \cap \partial H = \emptyset$. As above, this claim implies that $H'$ is nested in $H$. Let $I$ be the boundary of $H'$ after a clockwise-rotation around~$v$, and let $J$ be the boundary after a counter-clockwise rotation around~$v$. Assume by contradiction that $I$ and $J$ both intersect $\partial H$. See \autoref{fig:use cones}. Cut $I$ into two semi-infinite walks at $v$, the right part (with respect to $g$) denoted by $I_0$, and the left part denoted by $I_1$. Cut $J$ into two parts at $v$, the right part $J_0$, the left part $J_1$. The concatenation of $P$ and $I_1$ is a reduced walk, so $I_1$ is disjoint from $\partial H$ by \autoref{redux walk}, and so $I_0$ intersects $\partial H$ in a vertex $a$. The concatenation of $P$ and $J_0$ is a reduced walk, so $J_0$ is disjoint from $\partial H$ by \autoref{redux walk}, and so $J_1$ intersects $\partial H$ in a vertex $b$. Let $Q'$ be the concatenation of the subpath of $I_0$ between $a$ and $v$, and of the subpath of $J_1$ between $v$ and $b$. Let $Q$ be the subpath of $\partial H$ between $a$ and $b$. Then $Q$ and $Q'$ are distinct reduced paths with the same end-vertices in $T$, contradicting \autoref{redux walk}.
\end{proof}

\begin{lemma}\label{lemescape}
Let $T$ be a reducing triangulation of the plane. Let $G$ be a graph, and let $f : G \to T^1$ be simplicial. Let $H$ be a cone of $T$, and let $v$ be a vertex of $G$. If $f$ is strongly harmonious, and if $f(v)$ is the tip of $H$, then there is a walk $W$ based at $v$ in $G$ that satisfies $f(W) \subset H \cup \partial H$ and $f(W) \not\subset \partial H$.
\end{lemma}

\begin{proof}
Let $G_0 \subset G$ be the cluster of $f$ containing $v$. Since $f$ is strongly harmonious, there is a directed edge $e$ from $G_0$ to $G \setminus G_0$ such that $f(e) \in H \cup \partial H$. If $f(e) \in H$, then we are done. So we assume $f(e) \in \partial H$. The boundary $\partial H$ splits into two semi-infinite walks $I_0$ and $I_1$ based at $f(v)$, where $H$ lies on the left of $I_0$ and on the right of $I_1$. Assume $f(e) \in I_0$, the other case being similar. Let $w$ be the head vertex of $e$. Let $J$ be the suffix of $I_0$ obtained by removing its first edge ($f(w)$ is the first vertex of $J$). Let $L$ be the bi-infinite walk that makes only $3_r$-turns and contains $J$ as a subwalk. Since $f$ is strongly harmonious, there is a walk $W$ based at $w$ in $G$, such that $f \circ W$ may stay in $J$ for a while, and then leaves $J$ on its left.
\end{proof}

In order to ease the reading of the proof of \autoref{prop:tutteinplane}, we shall not use \autoref{lempush} and \autoref{lemescape} as they are. Instead, we shall use two corollaries:

\begin{corollary}\label{corplane1}
Let $T$ be a reducing triangulation of the plane. Let $G$ be a graph, and let $f : G \to T^1$ be simplicial. Let $H$ be a half-plane of $T$, and let $v$ be a vertex of $G$. Let $\kappa > 0$. If $f$ is strongly harmonious, and if $f(v) \in H$, then there is a walk $W$ from $v$ to a vertex $w$ in $G$, such that $f(W) \subset H$, and such that $f(w)$ is at distance greater than $\kappa$ from $\partial H$.
\end{corollary}

\begin{proof}
  This follows immediately from repeated iterations of Lemmas~\ref{lempush} and~\ref{lemescape}.
\end{proof}

\begin{figure}
    \centering
    \includegraphics[scale=1]{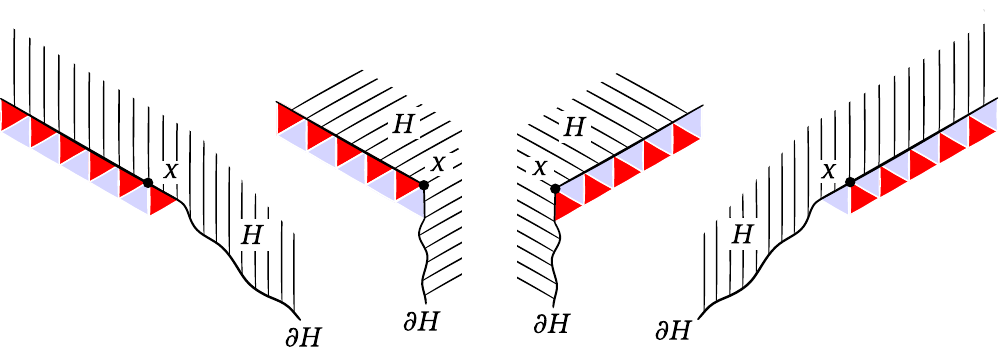}
    \caption{Combs.}
    \label{fig:combs}
\end{figure}

The next corollary is easier to state with the following definition. See \autoref{fig:combs}. In a reducing triangulation $T$ of the plane, a right (resp.\ left) \emphdef{comb} is a pair $(H,x)$ where $H$ is a half-plane of $T$, and $x$ is a vertex of $\partial H$, that satisfy the following: directing $\partial H$ so that $H$ lies on its right (resp.\ left), every turn of $\partial H$ at $x$ and after $x$ is $3_r$ or $2_b$ (resp.\ $-3_r$ or $-2_b$). The \emphdef{glen} of $(H,x)$ is the union of $H$ and of the part of $\partial H$ consecutive to $x$, including $x$.

\begin{corollary}\label{corplane2}
Let $T$ be a reducing triangulation of the plane. Let $G$ be a graph, and let $f : G \to T^1$ be simplicial. Let $(H,x)$ be a (left or right) comb of $T$, and let $v$ be a vertex of $G$. Let $\kappa > 0$. If $f$ is strongly harmonious, and if $f(v) = x$, then there is a walk $W$ from $v$ to a vertex $w$ in $G$, such that $f(W)$ is included in the glen of $(H,x)$, and such that $f(w)$ is at distance greater than $\kappa$ from $\partial H$.
\end{corollary}

\begin{proof}
The vertex $x$ is the tip of a cone $H'$ such that $H' \cup \partial H'$ is included in the glen of $(H,x)$. Lemmas~\ref{lempush} and~\ref{lemescape} conclude.
\end{proof}

We need a few more technical lemmas.

\begin{lemma}\label{balloon}
Let $T$ be a plane reducing triangulation, and let $T'$ be subgraph of $T^1$. If $T'$ is connected, finite, and not a single vertex, then $T'$ admits at least two vertices $x$ such that the edges of $T'$ incident to $x$ are all included in some bad turn of $T$.
\end{lemma}

\begin{proof}
If $T'$ has more than one vertex of degree one, then we are done. So assume that this is not the case. Then $T'$ has at least three vertices since it has neither loops nor multiple edges (\autoref{sec:reducing triangulations}). First assume that $T'$ has no degree one vertex. Then the outer closed walk $W$ of $T'$ has length greater than or equal to three, and never uses an edge of $T^1$ and its reversal consecutively. Orient $W$ so that the outer-face of $T'$ lies on the right of $W$. Walk along $W$ until some vertex $x_0$ is visited for the second time, then cut the portion of $W$ from $x_0$ to itself, thus obtaining a closed walk $C$ based at $x_0$. Then $C$ is a simple closed walk, not a single vertex. Thus, since $T^1$ has no loop nor multiple edges, $C$ has length greater than or equal to three. Also, every vertex $x \neq x_0$ of $C$ is such that all the edges of $T'$ incident to $x$ either lie on $C$, or on the left of $C$. \autoref{redux walk 2} ensures that at least three of the left turns of $C$ are bad, and at least two of them do not occur at $x_0$, proving the lemma in this case.

Now assume that $T'$ has exactly one vertex of degree one. Removing degree one vertices from $T'$ as long as there is one immediately gives the following: $T'$ is the union of a path $P$, not single vertex, and of a connected graph $Q$, not a single vertex, such that the intersection of $P$ and $Q$ is one of the two end-vertices $x$ of $P$, and such that $Q$ has no degree one vertex. By the preceding, $Q$ has a vertex distinct from $x$ that suits our need. And the end-vertex of $P$ distinct from $x$ also suits our needs.
\end{proof}

In the following, we denote by $T^0$ the set of vertices of a triangulation $T$.

\begin{lemma}\label{lem:disconnect}
Let $P$ be the plane. Let $T$ be a reducing triangulation of $P$. Let $Z$ be a triangulation of $P$, and let $\varphi : Z \to T$ be simplicial. If $\varphi$ is strongly harmonious, and if $\varphi$ is uniformly homotopic to the identity map of $P$, then $\varphi^{-1}(T^0)$ does not disconnect $P$.
\end{lemma}

\begin{proof}
By contradiction, assume that it does. There is no simple bi-infinite walk $W$ in $Z^1$ that belongs entirely to $\varphi^{-1}(T^0)$, for otherwise $\varphi(W)$ would be a single vertex of $T$ (being a connected subset of $T^0$), contradicting the assumption that $\varphi$ is uniformly homotopic to the identity map. Therefore, and since $\varphi$ is simplicial, there are a simple cycle $C$ in $Z^1$, a vertex $v \in Z^0$ in the interior of the bounded side of $C$, and a vertex $x$ of $T$, such that $\varphi(C) = x$ and $\varphi(v) \neq x$. Let $G$ be the subgraph of $Z^1$ that contains $C$ and the part of $Z^1$ lying in the bounded side of $C$. Then $\varphi(G)$ is a finite connected subgraph of $T$, not a single vertex. So by \autoref{balloon} there is a vertex $y \neq x$ in $\varphi(G)$ such that the set $B$ of edges of $\varphi(G)$ incident to $y$ is included in a bad turn of $T$. There is a cluster $G_0 \subset Z^1$ such that $\varphi(G_0) = y$, and such that the edges between $G_0$ and $Z^1 \setminus G_0$ all belong to $G$. Those edges are mapped to $B$ by $\varphi$, contradicting the assumption that $\varphi$ is strongly harmonious.
\end{proof}

In the following, we say that $1$-turns and $2_r$-turns are \emphdef{bad left turns}.

\begin{lemma}\label{lem:bad left}
In a plane reducing triangulation $T$, let $I$ be a simple bi-infinite walk that does not make any bad left turn. The vertices of $T$ on the left of $I$ at distance one from $I$ are the vertices of a simple bi-infinite walk $I'$ that does not make any bad left turn.
\end{lemma}

\begin{proof}
Consider the sequence $E$ of directed edges of $T$ emanating from vertices of $I$ to the left side of $I$. If $e_1$ and $e_2$ are consecutive in $E$, then either $e_1$ and $e_2$ have the same target vertex, or their they have the same source vertex. In the first case, $e_2$ is the counter-clockwise rotation of $e_1$ around their target vertex, in the second case $e_2$ is the clockwise-rotation of $e_1$ around their source vertex. Let $I'$ be the bi-infinite walk corresponding to the target vertices of the directed edges in $E$. The vertices of $I'$ are those on the left of $I$ at distance one from $I$.

Then $I'$ does not make any $0$-turn, nor any bad left turn. For otherwise there are three consecutive $e_1, e_2, e_3 \in E$ that all have the same target vertex (at which $I'$ makes a bad left turn), such that $e_2$ sees red on its left. But then $I$ makes a $2_r$-turn at the source vertex of $e_2$, a contradiction.

Moreover $I'$ is simple. For otherwise some portion $W$ of $I'$ is a non-trivial simple closed walk, and the bounded side of $W$ lies on its right by \autoref{redux walk 2}, and since $W$ makes at most one bad left turn. But then $I$ is contained in the bounded side of $W$, a contradiction.
\end{proof}

\begin{figure}
    \centering
    \includegraphics[scale=1]{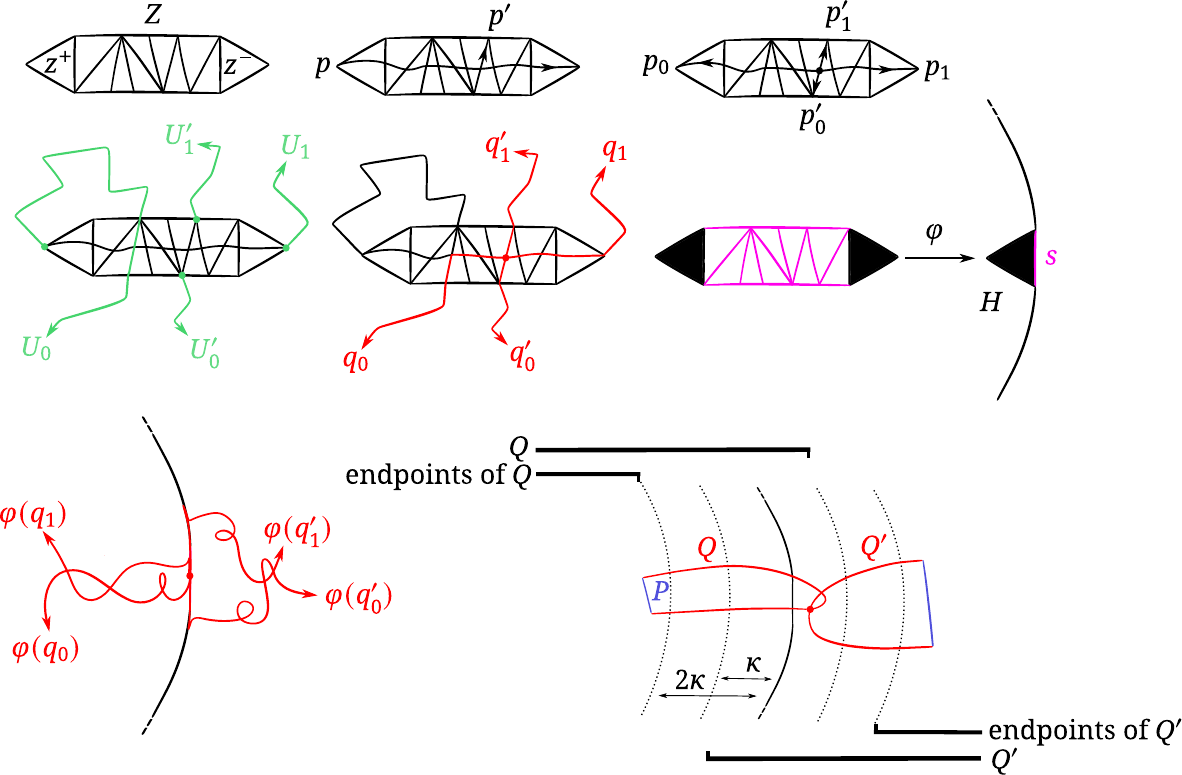}
    \caption{The construction in the proof of \autoref{prop:tutteinplane}.}
    \label{fig:core proof}
\end{figure}

\begin{proof}[Proof of Proposition~\ref{prop:tutteinplane}]
The overall proof is illustrated in \autoref{fig:core proof}.  By contradiction, assume that there are two faces $z_0$ and $z_-$ of $Z$ such that $\varphi|_{z_+}$ is positive and $\varphi|_{z_-}$ is negative. By \autoref{lem:disconnect}, there is a simple path $p$ from $z_+$ to $z_-$ in the dual of $Z$, that is disjoint from $\varphi^{-1}(T^0)$. Up to replacing $z_+$ and $z_-$ by other faces of $Z$, we may assume that $\varphi$ is null on every face of $Z$ intermediately visited by $p$. Extend $p$ to the respective vertices of $z_+$ and $z_-$ not incident to the first and last edges crossed by $p$.  Because $\varphi$ is null on every face visited by~$p$ except the first and last ones, and because $p$ is disjoint from~$\varphi^{-1}(T^0)$, there is an edge~$s\in T^1$ such that the image, by~$\varphi$, of all the edges of~$Z$ crossed by~$p$ is exactly~$s$.  Moreover, one can direct~$s$ in such a way that all the edges of~$Z$ crossed by~$p$ from left to right are mapped to~$s$.  Consider the left-most bi-infinite reduced walk of $T$ that contains $s$, and let $H$ be the half-plane on its left. Because $\varphi|_{z_+}$ is positive and $\varphi|_{z_-}$ is negative, the interiors $\mathring z_+$ and~$\mathring z_-$ of triangles $z_+$ and~$z_-$ are mapped inside~$H$.

Using the assumption that $\varphi$ is uniformly homotopic to the identity, there is $\kappa > 0$ such that each vertex of $T$ is at distance smaller than~$\kappa$ from its pre-images under~$\varphi$. Let $p'$ be any of the directed edges crossed by $p$ from left to right. By \autoref{corplane1} there is, in $Z^1$, a walk $U_0$ (resp.~$U_1$) based at the source (resp.\ target) end-vertex of~$p$ such that, for $i=1,2$, $\varphi(U_i) \subset H$, and such that $\varphi(U_i)$ contains a point at distance greater than $3 \kappa$ from $\partial H$.  Cut $p$ into two parts $p_0$ and $p_1$ at its intersection with $p'$, and reverse $p_0$. Let $q_i$ be the concatenation of $p_i$ and $U_i$. Make $q_i$ a simple path by shortening it if necessary. Then $q_0$ and $q_1$ are disjoint from $p'$ except for their basepoint, since $\varphi(U_0), \varphi(U_1)\subset H$, $\varphi(p')=s$, and $\varphi(\mathring z_+),\varphi(\mathring z_-)\subset H$. We prove by contradiction that $q_0$ and $q_1$ are disjoint except for their basepoint. If not, then $q_0 \cup q_1$ contains a simple closed curve whose bounded side contains an endpoint~$y$ of~$p'$. By \autoref{corplane2} there is a walk $X$ based at $y$ in $Z^1$ such that $\varphi(X)$ is disjoint from $\mathring{s} \cup H$, and such that $X$ intersects $q_0 \cup q_1$. That contradicts the fact that $\varphi(q_0 \cup q_1) \subset \mathring{s} \cup H$.

Now cut $p'$ into two parts $p'_0$ and $p'_1$ at its intersection with $p$. Reverse $p'_0$. By \autoref{corplane2} there is a simple walk $U'_0$ (resp.\ $U'_1$) based at the source (resp.\ target) end-vertex of~$p'$ in $Z^1$, such that $\varphi(U'_i)$ is disjoint from $\mathring{s} \cup H$, and such that $\varphi(U'_i)$ contains a point at distance greater than $3 \kappa$ from $\partial H$. Let $q'_i$ be the concatenation of $p'_i$ and $U'_i$. Since $\varphi(U'_0)$ and $\varphi(U'_1)$ are disjoint from $\mathring{s} \cup H$, the paths $q'_0$ and $q'_1$ are simple, and each of them is disjoint from $q_0$ and $q_1$ except for their basepoint. An argument by contradiction similar to the one of the previous paragraph shows that $q'_0$ and $q'_1$ are disjoint except for their basepoint.

Let $Q$ be the concatenation of~$q_0$ and~$q_1$, and $Q'$ be the concatenation of~$q'_0$ and~$q'_1$.  By construction, $Q$ and~$Q'$ are simple, and cross exactly once. Until now, we have considered separately the situation in~$Z$ (before applying~$\varphi$) and in~$T$ (after applying~$\varphi$).  But recall that $Z$ and~$T$ are both triangulations of the plane, and that each vertex of $T$ is at distance smaller than~$\kappa$ from its pre-images under~$\varphi$.  Because of this, and by the properties of~$\varphi(Q)$ and~$\varphi(Q')$, we have that $Q$ has its endpoints inside~$H$, at distance at least~$2\kappa$ from~$\partial H$, and may enter the complement of~$H$, but only by a distance of at most~$\kappa$.  Similarly, $Q'$ has its endpoints outside~$H$, at distance at least~$2\kappa$ from~$\partial H$, and may enter~$H$ but only by a distance of at most~$\kappa$. By (repeated applications of) \autoref{lem:bad left}, the vertices of $T$ inside $H$ at distance $2 \kappa$ from $\partial H$ are the vertices of a simple bi-infinite walk $I$. Then $I$ separates the end-points of $Q$ from every point of $Q'$. Join the end-points of $Q$ by a path $P$, where $P$ is separated from $Q'$ by $I$, thus turning $Q$ into a (not necessarily simple) closed curve $C$. Because $Q$ and~$Q'$ cross exactly once, and because $P$ and $Q'$ are disjoint, then $C$ and $Q'$ cross exactly once. Similarly, there is a simple bi-infinite walk that separates the end-points of $Q'$ from $C$, so $Q'$ can be extended to a closed curve $C'$ such that $C$ and $C'$ cross exactly once. But this is impossible, since any two closed curves in the plane cross an even number of times.
\end{proof}

\subsection{Proof of \autoref{tutte theorem}}

We now have almost all the material to prove \autoref{tutte theorem}. We first prove the theorem for strongly harmonious drawings, in the following proposition.

\begin{proposition}\label{weak tutte}
Let $S$ be a closed surface distinct from the sphere. Let $T$ be a reducing triangulation of $S$. Let $G$ be a finite graph embedded in $S$, and let $f: G \to T^1$ be simplicial. If $f$ is strongly harmonious, and if $f$ is homotopic to the inclusion map $G \hookrightarrow S$, then $f$ is a weak embedding.
\end{proposition}

(As a side note, observe that \autoref{weak tutte} considers also the torus.)

\begin{proof}
By \autoref{prop:extension} there are a triangulation $Z$ of $S$ whose 1-skeleton contains a subdivision of $G$ as a subgraph, and a simplicial map $\varphi : Z \to T$ with $\varphi|_G^{T^1} = f$, such that $\varphi$ is strongly harmonious and homotopic to the identity map of $S$. By \autoref{prop:tighttriang} there cannot be two faces $z_+$ and $z_-$ of $Z$ for which $\varphi|_{z_+}$ is positive and $\varphi|_{z_-}$ is negative. Thus, letting $Y \subset S$ contain one point in the interior of each face of $T$, $\varphi$ is coherently oriented at $Y$. By \autoref{prop:selfmaps} the map $\varphi|_{S \setminus \varphi^{-1}(Y)}^{S \setminus Y}$ is homotopic to a homeomorphism $S \setminus \varphi^{-1}(Y) \to S \setminus Y$. In particular, and since $G \cap \varphi^{-1}(Y) = \emptyset$, the map $f|^{S \setminus Y}$ is homotopic to an embedding $G \to S \setminus Y$. Also, $S \setminus Y$ is the patch system of $T$, and $f$ has no spur since $f$ is strongly harmonious. So $f$ is a weak embedding by \autoref{tutteinpatch}.
\end{proof}

The proof of \autoref{tutte theorem} relies on a few additional lemmas. 

\begin{lemma}\label{curves weak embedding}
Let $S$ be a closed surface distinct from the sphere and the torus. Let $T$ be a reducing triangulation of $S$. Let $C$ be a collection of closed walks in $T$. If the walks in $C$ are reduced, and if $C$ is homotopic to an embedding, then $C$ is a weak embedding.
\end{lemma}

The proof of \autoref{curves weak embedding} adapts arguments from the proofs of~\cite[Proposition~4.2]{de2024untangling}and~\cite[Proposition~1.5]{dubois2024making}, though it is considerably simpler due to the fact we consider a collection of closed curves, homotopic to an embedding (instead of general graph drawings, and instead of curves with arbitrary geometric intersection number).

\begin{proof}
Without loss of generality the walks in $C$ are not single vertices, and are thus not contractible by \autoref{redux walk}. Let $\Sigma$ be the patch system of $T$. We claim that there is a simple collection of closed curves homotopic to $C$ in $\Sigma$. This claim implies the lemma by \autoref{tutteinpatch} since $C$ has no spur. To prove the claim consider a collection of closed curves $\Gamma$ homotopic to $C$ in $\Sigma$, and self-crossing as few times as possible in $\Sigma$. We prove the claim by contradiction so assume that $\Gamma$ self-intersects.  Let $\widetilde S$ be the universal covering space of $S$. Let $\widetilde \Sigma$ and $\widetilde T$ be the respective lifts of $\Sigma$ and $T$ in $\widetilde S$. Since $\Gamma$ self-intersects while being homotopic to a simple collection of closed curves, there are either a non-simple closed curve in $\Gamma$ that could be made simple by homotopy, or there are two simple closed curves in $\Gamma$ that cross while they could be made simple alltogether by homotopy. There are two cases. 

First assume that in $\widetilde S$ some lift of a curve from $\Gamma$ self-intersects. Then some portion of this lift is a simple loop $\ell$ in $\widetilde S$, based at the intersection point. If the bounded side $D$ of $\ell$ does not contain any face of $\widetilde T$, equivalently if $D \subset \widetilde \Sigma$, then the self-intersection of $\gamma$ can be removed by homotopy in $\Sigma$, contradicting the assumption that $\Gamma$ self-crosses as few times as possible. Otherwise, consider the walk $W$ in $\widetilde T$ that encodes the sequence of crossings of $\ell$ with the arcs of $\widetilde \Sigma$. Then $W$ is not a single vertex, the two end-vertices of $W$ are the same vertex, and $W$ is reduced, contradicting \autoref{redux walk}.

Now assume that in $\widetilde S$ every lift of every curve in $\Gamma$ is simple. There are two simple lifts $\widetilde \gamma_0$ and $\widetilde \gamma_1$ of curves from $\Gamma$ that intersect at least twice (see e.g.~\cite[Theorem~3.5]{hass1985intersections}). Then some portion $\ell_0$ of $\widetilde \gamma_0$ and some portion $\ell_1$ of $\widetilde \gamma_1$ are such that $\ell_0$ and $\ell_1$ are the two sides of a bigon $D$ embedded in $\widetilde S$, between two lifted self-intersections of $\gamma$. If $D \subset \widetilde \Sigma$, then the two self-intersections of $\gamma$ can be removed by homotopy in $\Sigma$, contradicting the assumption that $\gamma$ self-crosses as few times as possible. Otherwise, consider the walks $W_0$ and $W_1$ in $\widetilde T$ that encode the sequence of crossings of $\ell_0$ and $\ell_1$ with the arcs of $\widetilde \Sigma$. Then $W_0$ and $W_1$ are distinct, have the same end-vertices, and are reduced, contradicting \autoref{redux walk}. 
\end{proof}

Preparing for the next lemma, observe that if $M$ is a finite graph embedded on a surface, then the universal cover of $M$ is a tree $\widetilde M$ naturally equipped with a rotation system. Moreover, if $C_0$ is a closed walk without spur in $M$, then every lift $\widetilde C_0$ of $C_0$ partitions $\widetilde M$ into three parts: the left of $\widetilde C_0$, the right of $\widetilde C_0$, and the image of $\widetilde C_0$.

\begin{lemma}\label{new curve}
Let $M$ be a finite graph embedded on a surface. Let $G$ be a finite graph simplicially mapped to $M$, and let $C$ be a collection of closed walks in $M$. Assume that $G$ and $C$ are weak embeddings without spur, and that their union is not a weak embedding. Then in the universal cover of $M$ there is a lift of $G$ that has vertices on both sides of a lift of a walk $C_0$ from $C$.
\end{lemma}

\begin{proof}
Without loss of generality every edge of $G$ is mapped to an edge of $M$ (contract the clusters of $G$ otherwise). Let $\Sigma$ be the patch system of $M$, and let $G'$ and $C'$ be embeddings approximating $G$ and $C$ in $\Sigma$. Without loss of generality $G'$, $C'$, and the arcs of $\Sigma$ are in general position. Also, $G'$ and $C'$ cross minimally. By assumption, they cross at a point $x$ in the interior of some face of $\Sigma$. Let $\widetilde \Sigma$ be the universal covering space of $\Sigma$, and let $\widetilde x$ be a lift of $x$ in $\widetilde \Sigma$. Let $\widetilde C_0'$ be the lift of a walk $C'_0$ from $C'$ that contains $\widetilde x$. Cut the lifts of $G'$ at every intersection with $\widetilde C'_0$, and let $\widetilde G'$ be one of the two cuts that meet $\widetilde x$. Since the number of crossings between $G'$ and $C'$ is minimal, $\widetilde G'$ is not contained in the union of the faces and arcs of $\widetilde \Sigma$ used by $\widetilde C'_0$.
\end{proof}

The following is analogous to \autoref{lem:planar part}.

\begin{lemma}\label{map annulus}
Let $S$ be a surface. Let $G$ be a graph, and let $f : G \to S$ be a map. If $f(G)$ is a simple circle $S_0 \subset S$, if $S_0$ does not bound a disk in $S$, and if $f$ is homotopic to an embedding in $S$, then $f$ is homotopic to an embedding in a tubular neighborhood of $S_0$.
\end{lemma}
\begin{proof}

Without loss of generality $G$ is connected. Fix a vertex $r$ of~$G$ and a spanning tree $T$ of $G$ rooted at~$r$.  Let $\ell$ be the simple loop on~$S$ based at~$x:=f(r)$ whose image is~$S_0$.  We first claim that we can assume without loss of generality that each edge of~$T$ is mapped to $x$ by $f$, and that each edge not in~$T$ is mapped, under~$f$, to a power of $\ell$. To see this, contract, in~$S_0$, the edges in~$T$; each remaining edge is homotopic to some power of~$\ell$ in $S_0$. We second claim that each edge of $T$ is actually mapped to either $\ell$, to $\ell^{-1}$, or to the constant loop in $S_0$. This is due to a result of Epstein~\cite[Theorem~4.2]{e-c2mi-66}, and since the edges not in $T$ are mapped to loops that can be made simple by homotopy in $S$.

On the other hand, there is an embedding $f' : G \to S$ homotopic, in~$S$, to $f$. Without loss of generality, we can assume that this homotopy between $f$ and~$f'$ holds the image of~$r$ fixed. Indeed, consider a homotopy $H$ from $f$ to $f'$, and the path $p : [0,1] \to S$ followed by the image of $r$ under $H$. There is an ambiant isotopy $H'$ of $S$ that “counteracts” $H$ in the sense that it maps $p(t)$ to $x$ for all $t \in [0,1]$. Composing the maps in $H$ by the maps in $H'$ gives a homotopy from $f$ to an embedding (not $f'$) that helds the image of $r$ fixed.

We contract the edges of~$T$ in~$f'$ to a small neighborhood of~$x$, this time preserving the fact that we have an embedding.  The remaining edges are loops that, under~$f'$, are homotopic to their counterparts under~$f$ (if $T$ would be really contracted to~$x$). The contractible ones can be pushed by isotopy into a small neighborhood of $x$, as each of them bounds a disk with only (possibly) contractible loops inside it~\cite[Theorem~1.7]{e-c2mi-66}. The other loops can be bundled together parallel, since any pair of them bounds a disk with only (possibly) contractible or homotopic loops inside it. Then they can be pushed alltogether in a neighborhood of $\ell$.
\end{proof}

\begin{proof}[Proof of Theorem~\ref{tutte theorem}]
Clearly if $f$ is a weak embedding, then there is an embedding homotopic to $f$ in $S$. For the other direction assume that there is an embedding homotopic to $f$ in $S$. We shall prove that $f$ is a weak embedding.

Partition $G$ into two subgraphs $A$ and $B$ such that $f|_B$ is strongly harmonious, and $f$ is not strongly harmonious on any of the connected components of $A$. Then $f|_{A} = C \circ f'$ for some disjoint union $O$ of cycle graphs, mapped to reduced closed walks by $C : O \to T^1$, and some simplicial map $f' : A \to O$, without spur. 

Our first claim is that the collection of closed walks $C$ is homotopic to an embedding in $S$. Indeed, $C$ can be realized as the restriction of $f$ to a collection of disjoint cycles in $A$, as follows. For every cycle $O_0$ in $O$, since $f'$ has no spur, there is a simple closed walk $W$ in $A$ mapped by $f'$ to a non-trivial power of $O_0$. Then $W$ is mapped by $f$ to a non-trivial power of the closed walk $C_0 := C|_{O_0}$, homotopic to an embedding (since $f$ is), and so it is actually mapped to $C|_0$ or its reversal~\cite[Theorem~4.2]{e-c2mi-66}.

Thus, $C$ is a weak embedding by \autoref{curves weak embedding}, and since the walks in $C$ are reduced. Now $f'$ is a weak embedding by \autoref{map annulus} and \autoref{tutteinpatch}. Our second claim is that $f|_{B} \cup C$ is a weak embedding. This claim proves the theorem as $f|_{A \cup B}$ is then a weak embedding. We prove the claim by contradiction so assume that $f|_{B} \cup C$ is not a weak embedding. Let $U_T$ be the universal cover of \emph{the 1-skeleton} of $T$. There is by \autoref{new curve} some lift of $f|_B$ in $U_T$ that contains vertices on both sides of some lift of a walk $C_0$ from $C$. Now let $\widetilde S$ be the universal cover of $S$. In $\widetilde S$, there are lifts $\widetilde B$ and $\widetilde C_0$ of $f|_B$ and $C_0$ such that $\widetilde B$ contains vertices on both sides of $\widetilde C_0$. Then $\widetilde C_0$ is a semi-infinite reduced walk in the lift of $T$, and $\widetilde B \cup \widetilde C_0$ is uniformly homotopic to an embedding. That contradicts \autoref{corplane1}.
\end{proof}

\section{Harmonizing a drawing monotonically: Proof of \autoref{harmonization theorem}}\label{sec:harmonizing}

In this section we prove \autoref{harmonization theorem}, which we restate for convenience:

\harmonizationtheorem*

The strategy is to transform a drawing~$f$ iteratively by some ``local'' moves satisfying the property that, if no move can be applied, then the current drawing is harmonious.  Two of these moves, the shortening and balancing moves, decrease strictly the length of the drawing.  On the other hand, the flip move does not affect the length of the drawing.  Roughly but not exactly, the algorithm performs moves iteratively, giving priority to the shortening and balancing moves over the flip moves.  To bound the complexity of the algorithm, it then suffices to bound the length of a flip sequence, assuming that at each step of this sequence, no shortening or balancing move is possible.  The details are actually more complicated; in particular, the flip sequence is chosen carefully.

\subsection{Reduction to simplicial maps}

We have the following preliminary observation:
\begin{lemma}\label{lem:simplicial}
  To prove \autoref{harmonization theorem}, we can assume that $f$ is simplicial.
\end{lemma}
\begin{proof}
  Let $f:G\to T^1$ be as in the statement of \autoref{harmonization theorem}.  Let $\bar f:\bar G\to T^1$ be the associated simplicial map that factorizes~$f$.  If \autoref{harmonization theorem} holds for simplicial maps, then we obtain a map $\bar f':\bar G\to T^1$ that is harmonious, is homotopic to~$\bar f$, and does not increase the length of the edges (compared to~$f'$).  In particular, $\bar f'$ is simplicial.  It naturally corresponds to a drawing~$f'$ of~$G$ on~$T^1$ that satisfies the desired properties.
\end{proof}

\subsection{Flips, shortenings, and balancings}

\begin{figure}
    \centering
    \includegraphics[scale=1]{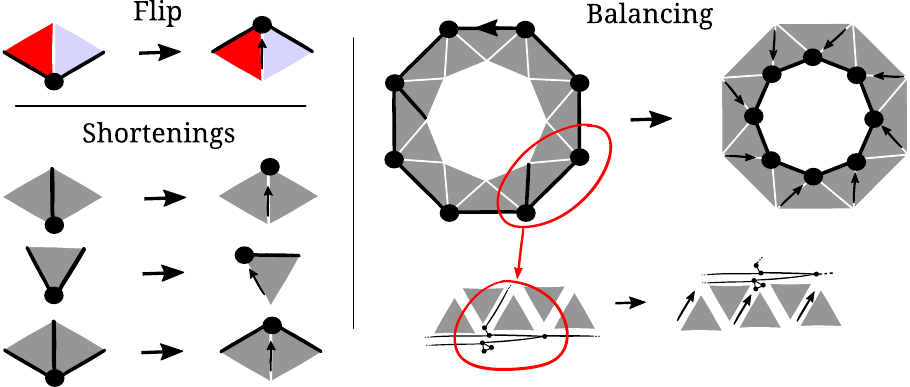}
    \caption{The moves. (Left) Each disk represents a single cluster. (Right) Each disk may represent several clusters. The balancing is slightly counter-clockwise (clockwise would not decrease any edge length here).}
    \label{fig:moves}
\end{figure}

Throughout this section, $T$ is a reducing triangulation of a surface $S$ distinct from the sphere and the torus, $G$ is a graph, and $f : G \to T^1$ is a \emph{simplicial} map.  Recall that $f$ factors uniquely into a homomorphism~$\hat f:\hat G\to T^1$.  We now define the three moves bringing $\hat f$ (and thus $f$) closer to harmony; see \autoref{fig:moves}.

\begin{itemize}
  \item First, if the edges of~$\hat G$ incident with~$v$ leave~$v$ via two edges around $\hat f(v)$, which together form a $2_r$-turn around~$\hat f(v)$, then we can perform a \emphdef{flip move} to~$\hat f$, which transforms $\hat f$ into a homotopic map $\hat f':\hat G\to T^1$ (\autoref{fig:moves}, top left), which is actually also a homomorphism.  From $\hat f'$, we immediately deduce a simplicial map~$f':G\to T^1$.

  Since $\hat f, \hat f':\hat G\to T^1$ are both homomorphisms, we can also view a flip as a specific operation that turns a homomorphism into another one (we will use this point of view later).
  
  \item Second, let $v$ be a vertex of $\hat G$.  If the edges of~$\hat G$ incident with~$v$ leave~$v$ via one, two, or three consecutive edges of~$T^1$ around $\hat f(v)$, then we can perform a \emphdef{shortening move} to~$\hat f$, which transforms $\hat f$ into a homotopic, simplicial map~$\hat f'$ of~$\hat G$ in which no edge of~$\hat G$ is longer than in~$\hat f$ (\autoref{fig:moves}, bottom left).  From $\hat f'$, we immediately deduce a simplicial map~$f':G\to T^1$.

  \item Third, consider a simple directed cycle~$C$ in~$\hat G$ that makes only 3-turns under~$f$.  We say that a walk of~$\hat G$, identified by its sequence of directed edges $(e_0,\ldots,e_k)$, \emph{follows} $C$ if there is a walk on~$C$, its \emph{following walk}, (which may go back and forth on~$C$), identified by its directed edges $(c_0,\ldots,c_k)$ such that $f(e_i)=f(c_i)$ for each~$i$.  A walk $(e_0,\ldots,e_k,e_{k+1})$ of~$\hat G$ \emph{pulls $C$ to the left} if $(e_0,\ldots,e_k)$ follows~$C$, with following walk $(c_0,c_1,\ldots,c_k)$, and if moreover edge $f(e_{k+1})$ leaves the image of the directed cycle~$C$, at the end-vertex of~$c_k$, to its left.  The notion of being pulled to the right is defined analogously.

  If $C$ makes only 3-turns under~$f$, is pulled left, and is not pulled right, we can define a \emphdef{balancing move} as follows (\autoref{fig:moves}, right).  We consider all the vertices of~$\hat G$ that are parts of walks  following~$C$ (including the vertices of~$C$ themselves) and move them to the left of~$C$, in such a way that we still have a homomorphism when restricted to the vertices following~$C$.  There are exactly two possibilities to do this, depending on whether the cycle is rotated ``slightly clockwise'' or ``slightly counterclockwise'' (as in \autoref{fig:moves}).  In any case, the length of the image of an edge of~$\hat G$ does not increase, and we choose a possibility that strictly decreases the length of the image of at least one edge, based on how a walk pulls $C$ to the left.  As before, from $\hat f'$, we immediately deduce a simplicial map~$f':G\to T^1$.
\end{itemize}

These three moves are useful in the following sense:
\begin{lemma}\label{no move implies harmonious}
  If $f$ cannot be flipped, shortened, or balanced, then $f$ is harmonious.
\end{lemma}

\begin{figure}
    \centering
    \includegraphics[scale=1]{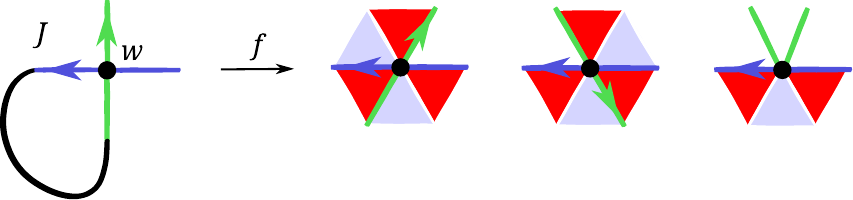}
    \caption{In the proof of \autoref{no move implies harmonious}, the first two and the last two edges of~$J$ cannot be mapped to two edge-disjoint walks.}
    \label{fig:cross}
\end{figure}

\begin{proof}
  Without loss of generality $G$ is connected. Also, $f$ is a homomorphism (for otherwise $f$ factors into a homomorphism $\hat f : \hat G \to T^1$ such that no move applies to $\hat f$ by definition, and such that if $\hat f$ is harmonious, then $f$ also). If $f$ is strongly harmonious, we are done, so we can assume that $f$ is not strongly harmonious.  So let $v$ be a vertex of~$G$, and let $L$ be a (left or right) line in~$T^1$ whose central vertex is~$f(v)$ such that, for each walk~$W$ in~$G$ based at~$v$, the image of~$W$ under~$f$ does not escape~$L$.
  In the universal cover, this means the following.  There are a vertex~$\widetilde x$ of~$\widetilde T^1$, projecting to $f(v)$, and a (left or right) line $\widetilde L$ in~$\widetilde T^1$ whose central vertex is~$\widetilde x$, such that, for each walk~$W$ in~$G$ based at~$v$, the lift of~$f \circ W$ based at~$\widetilde x$ does not escape~$\widetilde L$.  We assume that $\widetilde L$ is a left line (equivalently, that it makes only $3_r$-turns), the other case being similar.  We claim that, for each walk~$W$ in~$G$ based at~$v$, the lift of~$f \circ W$ based at~$\widetilde x$ is actually \emph{contained} in $\widetilde L$.

\smallskip

  First, we explain why the claim implies the lemma.  Because $T^1$ is finite, there are a cycle graph $O$ and a reduced closed walk $C : O \to T^1$ such that $\widetilde L$ is a lift of $C$. Using the claim, and since $\widetilde L$ is simple by \autoref{redux walk}, there is a simplicial map $f' : G \to O$ such that $f = C \circ f'$.  Moreover, $f'$ has no spur since $f$ cannot be shortened.  So $f$ is harmonious, proving the lemma.  
  There remains to prove the claim, which we do in the remaining part of the proof.

\smallskip

Recall that~$\widetilde L$ makes only $3_r$-turns. In~$G$, one can build a semi-infinite walk $I$ based at $v$ such that the lift of $f \circ I$ based at~$\widetilde x$ is equal to the non-negative part of~$\widetilde L$. (Indeed, at a given step on~$\widetilde L$, there is no edge that goes strictly to the right of~$\widetilde L$ since no walk based at $v$ can escape $L$ under $f$; if all edges go strictly to the left of~$\widetilde L$ or backward on~$\widetilde L$, a shortening or a flip could be applied to $f$.)  Since $I$ is a semi-infinite walk in~$G$, it contains a subwalk~$J$ that, after removing its first and last edge, becomes a simple loop~$Q$ in~$G$, based at some vertex~$w$ of~$G$. Let $P$ be the prefix of $I$ leading to $Q$. The lift of $f \circ P$ based at $\widetilde x$ is a portion of the non-negative part of $\widetilde L$, ending at a vertex $\widetilde y$ of $\widetilde T^1$. Also, $f \circ Q$, regarded as a closed walk by concatenating it with itself, makes a $3_r$-turn also at the middle occurrence of~$w$; Otherwise, the first two and the last two edges of~$J$ would map, under~$f$, to two edge-disjoint walks of length two making $3_r$-turns at~$f(w)$ (\autoref{fig:cross}), and so one could stop~$I$ at this point and escape from~$\widetilde L$, a contradiction. In particular the lift of $f \circ Q$ based at $\widetilde y$ is (a translate of) $\widetilde L$.

We conclude in two steps. First we prove that for every walk $W$ based at $v$ in $G$, the lift of $f \circ W$ based at $\widetilde x$ cannot enter the right side of $\widetilde L$, even after staying in $\widetilde L$ for a while. By contradiction, assume that it does. For every $n \geq 1$ the walk $P \cdot Q^n \cdot P^{-1} \cdot W$ is based at $v$ in $G$, and is such that the lift of $f \circ (P \cdot Q^n \cdot P^{-1} \cdot W)$ based at $\widetilde x$ enters the right side of $\widetilde L$ after staying in $\widetilde L$. There is $n$ such that this lift does not intersect $\widetilde L$ outside its non-negative part. We obtained a walk based at $v$ that escapes $L$ under $f$, a contradiction.

Now we prove that the lift of $f \circ W$ based at $\widetilde x$ cannot enter the left side of $\widetilde L$, even after staying in $\widetilde L$ for a while. By contradiction, assume that it does. Without loss of generality, removing the last edge from $W$ gives a walk $W'$ such that the lift of $f \circ W'$ based at $\widetilde x$ is contained in $\widetilde L$. Then the lift of $f \circ (P^{-1} \cdot W')$ based at $\widetilde y$ is contained in $\widetilde L$, which lifts $f \circ Q$, and so $P^{-1} \cdot W'$ follows $Q$. Then $P^{-1} \cdot W$ pulls $Q$ to the left. Since $Q$ is pulled to the left, and since no balancing applies, $Q$ is pulled to the right. So there is a walk $U$ based at $w$ in $G$ such that the lift of $f \circ U$ based at $\widetilde y$ is contained in $\widetilde L$, except for its last edge that enters the right side of $\widetilde L$. Then the walk $P \cdot U$, based at $v$, is such that the lift of $f \circ (P \cdot U)$ based at $\widetilde x$ enters the right side of $\widetilde L$ after staying in $\widetilde L$ for a while. That contradicts the previous paragraph. 
\end{proof}

\subsection{Preliminaries on flip sequences}

By the preceding lemma, a natural strategy is to apply flips, shortenings, and balancings as much as possible until it is not possible anymore.  Shortenings and balancings strictly decrease the length of the map, so only finitely many such moves can be applied.  Most of the argument thus focuses on sequences of flips in which no shortening or balancing can be applied at any step.  Recall that flips transform a homomorphism into another one, and thus \emph{henceforth we consider maps from $G$ to~$T^1$ that are homomorphisms}.

Formally, a \emphdef{flip sequence} is a sequence of \emph{homomorphisms} $f_0, \dots, f_p : G \to T^1$ such that $f_{i+1}$ results from a flip of $f_i$ for every $0 \leq i < p$, and no shortening or balancing can be applied to any of $f_0,\ldots,f_{p-1}$.  We use the following conventions. Given $0 \leq i < p$ we call flip $i$ and abusively denote by $i$ the flip from $f_i$ to $f_{i+1}$. Given $0 \leq i \leq j \leq p$, we denote by $F_{i \to j}$ the flip sequence $f_i, \dots, f_j$. Given a flip sequence $F$, and a vertex $v$ of $G$, we denote by $F|v$ the walk performed in $T$ by the image of $v$ through the flips of $v$ in $F$.

The map $f$ induces a \emphdef{left-blue direction} of~$G$, obtained by directing each edge~$e$ of~$G$ in such a way that $f(e)$ has a blue triangle on its left.  Thus, $G$ becomes a digraph.  A \emphdef{source} in a digraph is a vertex that has no incoming edge.  We will use the following trivial observation repeatedly, without mentioning it explicitly: Each flippable vertex~$v$ is a source of its left-blue direction, and each flip reverses the direction of the edges incident to~$v$.  We need a series of easy lemmas.

\begin{lemma}\label{alternating flips}
  Let $F$ be a flip sequence. If $v$ and $w$ are two adjacent vertices in $G$, then $v$ is flipped in $F$ in-between any two flips of $w$.
\end{lemma}
\begin{proof}
  When $w$ is flipped, it is a source in its left-blue direction, and it can only become a source again after~$v$ is flipped.
\end{proof}

\begin{lemma}\label{source that cannot be flipped}
  Let $F$ be a flip sequence. If, before the first flip of~$F$, vertex~$v$ is a source in its left-blue direction, and $v$ cannot be flipped, then $v$ is not flipped in~$F$.
\end{lemma}
\begin{proof}
  Vertex~$v$ cannot be flipped before at least one of its neighbors is flipped, but no neighbor~$w$ of~$v$ can be flipped before $v$ is flipped, because $w$ is not a source.
\end{proof}

\begin{lemma}\label{oriented cycle that cannot be flipped}
  Let $F$ be a flip sequence.  If $C$ is a cycle in $G$ (not reduced to a single vertex) that is a directed cycle in its left-blue direction, then no vertex of $C$ is flipped in $F$.
\end{lemma}
\begin{proof}
  The first vertex of~$C$ that would be flipped would not be a source (in its left-blue direction) before the flip.
\end{proof}

\begin{figure}
    \centering
    \includegraphics{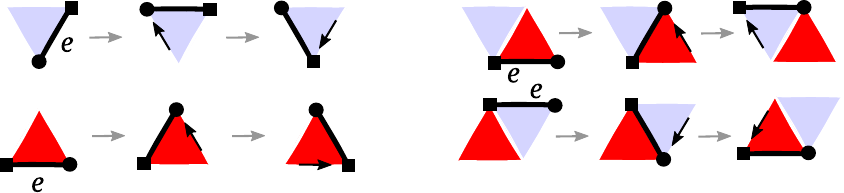}
    \caption{(Left) The second flip of the edge $e$ does not counteract the first flip of $e$. (Right) The second flip of $e$ counteracts the first one.}
    \label{fig:following}
\end{figure}

Let $e$ be an edge of $G$.  In a flip sequence, assume that flip~$i$ flips an end-vertex of $e$, that flip~$j$ flips the other end-vertex of $e$, and that the end-vertices of $e$ are not flipped between $i$ and $j$. We say that flip $j$ \emphdef{counteracts} flip $i$ if the image of $e$ is rotated clockwise by $i$ and counter-clockwise by $j$, or if it is rotated counter-clockwise by $i$ and clockwise by $j$.   See \autoref{fig:following}. 

\begin{lemma}\label{following 1}
Let $F$ be a flip sequence, and let $i<j$ be two flips of the same vertex~$v$ of~$G$, such that no flip of~$v$ appears between $i$ and~$j$.  Then we have the following properties:
\begin{itemize}
\item Every neighbor of~$v$ is flipped exactly once between $i$ and~$j$;
\item if every such flip counteracts flip~$i$, then $F_{i \to j+1}|v$ is a $3_r$-turn; otherwise, it is either a $1_r$-turn or a $-1_r$-turn.
\end{itemize}
\end{lemma}

\begin{cfigure}
    \centering
    \includegraphics{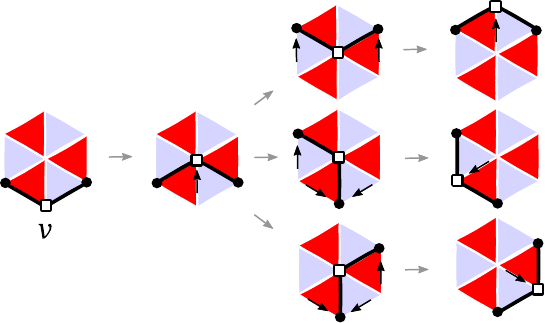}
    \caption{In a flip sequence, between two contiguous flips of a vertex $v$, the neighbors of $v$ are flipped once.}
    \label{fig:flip turns}
\end{cfigure}

\begin{proof}
By \autoref{alternating flips}, every neighbor of $v$ is flipped exactly once between $i$ and $j$, so there are three cases  depicted in \autoref{fig:flip turns}.
\end{proof}

\begin{lemma}\label{following 2}
Let $f_0,\ldots,f_p$ be a flip sequence.  If $i < j$ are flips of distinct adjacent vertices $v$ and $w$ respectively, and if no flip between $i$ and $j$ flips a neighbor of $w$, then flip $j$ counteracts flip~$i$.
\end{lemma}

\begin{cfigure}
    \centering
    \includegraphics{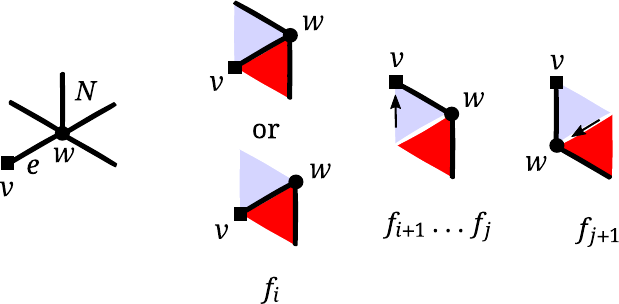}
    \caption{(Left) In the proof of \autoref{following 2}, the vertices $v$ and $w$, the edge $e$, and the set $N$ of edges incident to $w$. (Right) If flips $i$ and $j$ both rotate the image of $e$ clockwise, then $f_i$ can be shortened.}
    \label{fig:counteract}
\end{cfigure}

\begin{proof}
  Assume, for the sake of a contradiction, that flip $j$ does not counteract flip $i$. See \autoref{fig:counteract}. Let $e$ be the edge of $G$ between $v$ and $w$, directed from $v$ to~$w$. Assume that flips $i$ and $j$ rotate the image of $e$ clockwise, the other case being similar.  Let $N$ be the directed edges of~$G$ with source~$w$.

  We look at the situation in~$f_j$, and thus just before flip~$j$.  We have that $f_j$ maps~$N$ to two directed edges $a$ and~$b$ of~$T^1$ such that the reversal of~$a$, followed by~$b$, make a $2_r$-turn in~$T$.  Since $j$ rotates $e$ clockwise, $f_j(e)$ is the reversal of~$b$.
  
  Note that $w$ is not flipped between flips $i$ and~$j$ (because otherwise $v$, a neighbor of~$w$, would also be flipped between flips $i$ and~$j$, by \autoref{alternating flips}).  Since also no neighbor of $w$ is flipped between flips $i$ and $j$, we have $f_{i+1}(N) = f_j(N) = \{a,b\}$.  Since $i$ rotates the image of $e$ clockwise, $f_i(e)$ is the edge in the middle of the $2_r$-turn formed by $a$ and $b$ (directed towards the image of~$w$).  Thus, $f_i(N)$ is included in a set of three consecutive directed edges with source $f_i(w)$, and contains the middle directed edge, and so $f_i$ can be shortened, contradicting the fact that we have a flip sequence.
\end{proof}

\subsection{Proof of \autoref{harmonization theorem}}

In this section we prove \autoref{harmonization theorem}. The proof follows from a few definitions and lemmas.

\begin{lemma}\label{not flipping the root}
Assume that $G$ is connected and has $q$ vertices.  Let $r$ be a vertex of~$G$, and let $F$ be a flip sequence of~$G$ that never flips~$r$.  Then $F$ is composed of $O(q^2)$ flips.
\end{lemma}

\begin{proof}
In $G$ every vertex is at distance less than $q$ from $r$. By \autoref{alternating flips} every vertex at distance $i \geq 0$ from $r$ is flipped at most $i$ times.
\end{proof}

\begin{figure}
    \centering
    \includegraphics{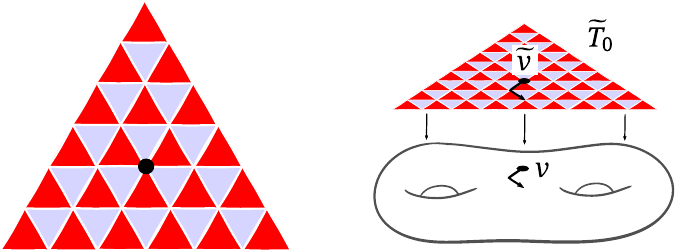}
    \caption{(Left) Flat zone of size six, and its central vertex. (Right) In the proof of \autoref{flat zone}, if there is a flat zone $\widetilde T_0$ of size nine in the universal covering triangulation of $T$, then the central vertex $\widetilde v$ of $\widetilde T_0$ projects to a vertex $v$ on $T$ such that every vertex of $T$ at distance two or less from $v$ has degree six.}
    \label{fig:flat zone}
\end{figure}

In a reducing triangulation $T$, a \emphdef{flat zone} of size $m' \geq 1$ is a subtriangulation $T_0$ of $T$ isomorphic to the subdivision of a triangle depicted in \autoref{fig:flat zone}, in which the sides of $T_0$ have length $m'$.

\begin{lemma}\label{flat zone}
  Every flat zone of the universal covering triangulation of $T$ has size less than $3(m+1)$, where $m$ is the number of edges of~$T^1$.
\end{lemma}
The key property used in the proof is that, as we have assumed, $S$ is not a torus.  This is actually the only place where this assumption is used.  (Since $T$ is a reducing triangulation, $S$ cannot be a sphere anyway.)
\begin{proof}
Assume the existence of a flat zone $\widetilde T_0$ of size $3(m+1)$ in the universal covering triangulation of $T$. See \autoref{fig:flat zone}. The central vertex of $\widetilde T_0$ is at distance $m+1$ from the boundary of $\widetilde T_0$. So every vertex of $T$ admits a lift in the interior of $\widetilde T_0$. Then every vertex of $T$ has degree six, so $S$ is a torus by Euler's formula, contradicting our assumption.
\end{proof}

A flip sequence $F$ is \emphdef{$k$-forward}, $k \geq 1$, if some vertex $v$ of $G$ is such that $F|v$ contains a subwalk of length~$k$ that makes only $3_r$-turns.

\begin{lemma}\label{straight implies flat zone}
Let $F$ be a flip sequence of homomorphisms $G \to T^1$. If $F$ is $k$-forward for some $k \geq 1$, then the universal covering triangulation of $T$ has a flat zone of size $k$.
\end{lemma}

\begin{proof}
Let $\widetilde T$ be the universal covering triangulation of $T$. Lift $F$ to a sequence $\widetilde F$ of homomorphisms $\widetilde G \to \widetilde T^1$, where $\widetilde G$ is a covering space of $G$. Assume without loss of generality (up to restricting to a smaller flip sequence) that some vertex~$v$ of~$\widetilde G$ is such that $\widetilde F|v$ has length~$k$ and makes only $3_r$-turns, and that the first and last flips of $\widetilde F$ are flips of $v$.  We shall prove by induction that $\widetilde F|v$ is a side of a flat zone in $\widetilde T$ lying to the left of $\widetilde F |v$. The case $k = 1$ is clear since a flat zone of size one is just a triangle of $\widetilde T$. So assume $k \geq 2$. Consider the first flip of $v$ in $\widetilde F$, and let $f : \widetilde G \to \widetilde T^1$ be the first map in $\widetilde F$. Let $e$ be the directed edge of $\widetilde T$ along which the flip is performed ($e$ is the first directed edge of $\widetilde F |v$). Consider the triangle $t$ of $\widetilde T$ on the left of $e$, and the vertex $x$ of $t$ that is not incident to $e$.

There is a neighbor $w$ of~$v$ in $\widetilde G$ such that $\widetilde f(w)=x$. By \autoref{alternating flips} there is in $\widetilde F$ a flip of $w$ in-between any two flips of $v$.  Recall that $v$ makes only $3_r$-turns; by \autoref{following 1}, this implies that every flip of $w$ counteracts the previous flip of $v$.  It follows that $\widetilde F |w$ is a walk of length $k-1$ that makes only $-3_r$-turns, running parallel to $\widetilde F |v$; by \autoref{following 1} again, these walks can only be $3_r$, $1_r$, or $-1_r$-turns, and because the degree of each vertex of~$T^1$ is at least six, the only possibility is that $\widetilde F|w$ makes only $3_r$-turns.  By induction, there is a flat zone of size $k-1$ on the left of $\widetilde F|w$, so there is a flat zone of side $k$ on the left of $\widetilde F |v$.
\end{proof}

We now need some definitions.  Consider an ordering $v_0,\ldots,v_{q-1}$ of the vertices of a graph.  We say that $v_i$ is a \emphdef{lowpoint} if every neighbor of~$v_i$, except perhaps~$v_0$, is higher than~$v_i$ in the ordering.  The ordering is \emphdef{proper} if (1) each vertex but~$v_0$ is adjacent to a lower vertex, and (2) the set of lowpoints has the form $\{0,\ldots,k\}$ for some~$k$.
A flip sequence $f_0, \dots, f_p$ is \emphdef{proper} if there is a proper ordering $v_0, \dots, v_{q-1}$ of the vertices of $G$ such that for every $0 \leq i < p$ the vertex flipped from $f_i$ to $f_{i+1}$ is $v_{i\bmod q}$.

\begin{lemma}\label{we go forward}
  Assume that $G$ is connected and has $q$ vertices. Let $F$ be a proper flip sequence of $kq+2$ homomorphisms $G \to T^1$,  for some integer~$k$.  Then $F$ is $k$-forward.
\end{lemma}
\begin{figure}
    \centering
    \includegraphics{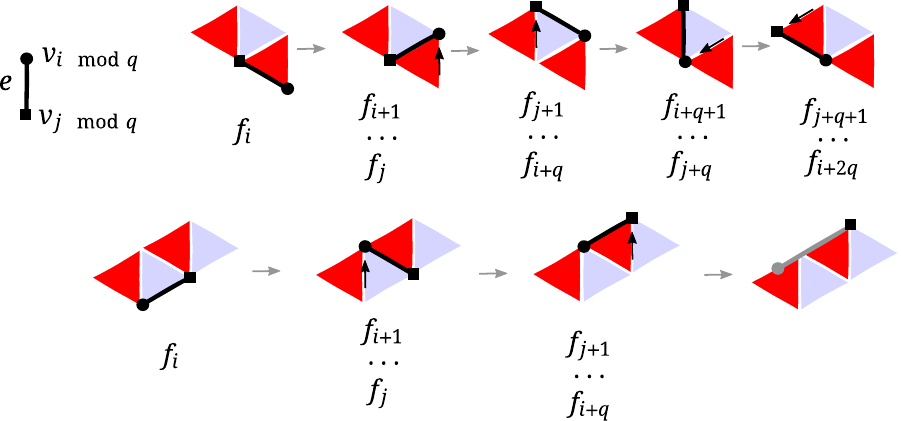}
    \caption{In the proof of \autoref{we go forward}, if $W_i$ makes a $1_r$-turn, then $W_j$ makes a $1_r$-turn.}
    \label{fig:flip follow twice}
\end{figure}

\begin{cfigure}
    \centering
    \includegraphics{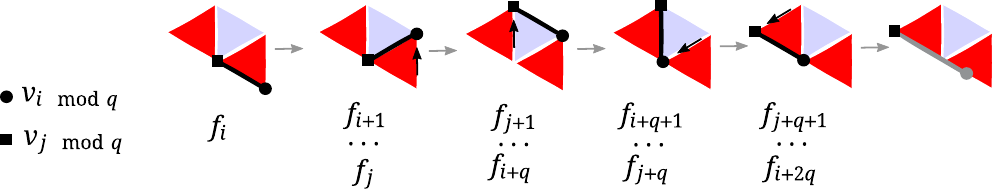}
    \caption{In the proof of \autoref{we go forward}, $W_i$ cannot make a $1_r$-turn.}
    \label{fig:flip contradiction}
\end{cfigure}

\begin{proof}
  Let $v_0,\ldots,v_{q-1}$ be the corresponding proper ordering of the vertices of~$G$.  Since $F$ is composed of $kq+1$~flips, the lowest vertex $v_0$ is both the first and the last vertex flipped in $F$, and $v_0$ is flipped $k + 1$ times in $F$, so $F|v_0$ has length $k+1$. Let $w_0$ be the highest neighbor of $v_0$. There are two cases. First assume that $w_0$ is a lowpoint.  Then, consider any neighbor~$x$ of~$v_0$.  Because $w_0$ is a lowpoint, $x$ is a lowpoint as well.  Thus, $v_0$ is the last neighbor of $x$ to be flipped in $F$ before a flip of $x$, and so every flip of $x$ counteracts the last flip of $v_0$ by \autoref{following 2}. That being true for every neighbor $x$ of $v_0$, \autoref{following 1} implies that $F|v_0$ makes only $3_r$-turns, and thus that $F$ is $(k+1)$-forward.

  Now assume that there is a neighbor of $w_0$ distinct from $v_0$ that is lower than $w_0$ in $L$.  Among the neighbors of $w_0$ lower than $w_0$, let $w_1$ be the highest. Among the neighbors of $w_1$ lower than $w_1$, let $w_2$ be the highest\dots, and so on until $w_{m-1} = v_0$ for some $m \geq 3$.  Then $w_{m-1}, \dots, w_0$ are the vertices of a directed cycle $C$ in $G$, in order, such that for every $i$ the vertex $w_{i+1}$ is both the last neighbor of $w_i$ flipped before $w_i$, and the last vertex of $C$ flipped before $w_i$, indices being taken modulo $m$.

  Let $I$ contain the flips $i \in \{0,\dots,(k-1)q\}$ whose vertex flipped, here $v_{i\bmod q}$, belongs to $C$. For every $i \in I$, the image of~$v_{i\bmod q}$ makes a walk $W_i$ of length two between $f_i$ and~$f_{i+q+1}$. Indeed, $v_{i\bmod q}$ is here flipped twice: in flips $i$ and $i+q$.

  We claim that for every two consecutive $i < j \in I$, if $W_i$ makes a $1_r$-turn, then $W_j$ also makes a $1_r$-turn. See \autoref{fig:flip follow twice}. Consider the edge $e$ between $v_{i \bmod q}$ and $v_{j \bmod q}$, and consider the flips $i$, $j$, $i+q$, and $j+q$. By definition no neighbor of $v_{j \mod q}$ is flipped between $i$ and $j$, nor between $i+q$ and $j+q$, so by \autoref{following 2} flip $j$ counteracts flip $i$, and flip $j+q$ counteracts flip $i+q$. Then flip $i$ cannot rotate the image of $e$ clockwise (as in the bottom part of \autoref{fig:flip follow twice}), for otherwise flip $j$ would rotate the image of $e$ counter-clockwise since it counteracts flip $i$, and then flip $i+q$ could not be such that $W_i$ makes a $1_r$-turn. So flip $i$ rotates the image of $e$ counter-clockwise (as in the top part of \autoref{fig:flip follow twice}). Then flip $j$ rotates the image of $e$ clockwise, since it counteracts flip $i$. Flip $i+q$ is also clockwise since $W_i$ makes a $1_r$-turn, and flip $j+q$ is counter-clockwise since it counteracts flip $i+q$. We proved that $W_j$ makes a $1_r$-turn.

  We use the claim immediately to prove by contradiction that for every $i \in I$, if $i \leq (k-2)q$, then $W_i$ does not make a $1_r$-turn. For otherwise, by our claim, every $l > i \in I$ is such that $W_l$ makes a $1_r$-turn. In particular, the smallest $j > i \in I$ is such that $W_j$ makes a $1_r$-turn, so the flips $i$, $j$, $i+q$, and $j+q$ are such as depicted in \autoref{fig:flip contradiction}. But then flip $i+2q$ cannot be such that $W_{i+q}$ makes a $1_r$-turn, which is a contradiction.

  The same arguments show that for every $i \in I$, if $i \leq (k-2)q$, then $W_i$ does not make a $-1_r$-turn. So $W_i$ makes a $3_r$-turn by \autoref{following 1}. In particular $F_{0 \to (k-1)q+1}|v_0$ makes only $3_r$-turns, so $F$ is $k$-forward.
\end{proof}

An ordering $v_0, v_1, \dots$ of the vertices of a digraph $D$ is \emphdef{monotonic} if every edge of $D$ is directed from a vertex $v_i$ to a vertex $v_j$ such that $i < j$.

\begin{lemma}\label{construct ordering}
  Let $D$ be a digraph of size $n$. If $D$ has no directed cycle and has a single source, then we can compute in $O(n)$ time an ordering of the vertices of $D$ that is both proper and monotonic.
\end{lemma}

\begin{proof}
  Without loss of generality, $D$ has at least two vertices. Let $v_0$ be the unique source of $D$. Since $D$ has no directed cycle, the vertex set of $D$ can be partitioned into sets $I_0, \dots, I_m$, $m \geq 1$, where $I_0 := \{v_0\}$, and where $I_k$ contains the sources of $D \setminus \bigcup_{i < k} I_i$ for every $1 \leq k \leq m$. Order arbitrarily each of the sets $I_0, \dots, I_m$, and concatenate them into an ordering $L$ of the vertices of $D$. Then $L$ is monotonic.  Moreover, each vertex but~$v_0$ is adjacent to a lower vertex, and $I_0\cup I_1$ is the set of lowpoints of~$L$, so $L$ is proper.
\end{proof}

Before proving \autoref{harmonization theorem} we detail how to detect, given a \emph{homomorphism} $f$, if a balancing move applies to $f$:

\begin{lemma}\label{find balancing}
Let $T$ be a reducing triangulation of size $m$. Let $G$ be a graph of size $n$, and let $f : G \to T^1$ be a homomorphism. In $O(m+n)$ time we can determine if there is a simple closed walk $C$ in $G$ such that $f \circ C$ makes only $3$-turns, is pulled left, and not right. In that case, both $C$ and the subgraph $G_0$ of $G$ spanned by the walks that follow $C$ are computed at the same time.
\end{lemma}

\begin{proof}
Without loss of generality, assume that we look for $C$ such that $f \circ C$ makes only $3_r$-turns (the $3_b$-turns case being symmetric). Consider the following algorithm. As a preliminary, build a graph $G'$ from $G$ by detaching every vertex $v$ of $G$ from its incident edges, and by re-attaching those edges to copies of $v$ as follows. Let $N$ contain the directed edges emanating from $v$ in $G$, whose basepoints have thus been detached from $v$. For every two directed edges $a$ and $b$ based at $f(v)$ in $T$ such that $(a^{-1}, b)$ makes a $3_r$-turn, consider all the directed edges in $N$  that are mapped to $a$ or $b$ (if any), and attach all their basepoints to a common copy $v'$ of $v$. Mark $v'$ with red if $f(N)$ contains a directed edge of $T$ on the right of $(a^{-1}, b)$. Mark $v'$ green if $f(N)$ does not contain any such directed edge, and if it contains a directed edge on the left of $(a^{-1}, b)$. 

The graph $G'$ \emph{projects} to $G$ in the sense that the edges of $G'$ are those of $G$, and every vertex of $G'$ corresponds to a vertex of $G$ (though not in a one-to-one manner). We say that a vertex $v'$ of $G'$ \emph{lifts} a vertex $v$ of $G$ if $v$ is the projection of $v'$.

Build the graph $G'$ in $O(m+n)$ time. Direct every edge $e'$ of $G'$ so that $f(e')$ sees red on its left. Then determine in $O(n)$ time if there is a connected component $G'_0$ of $G'$ that does not contain any red vertex, that contains a green vertex, and that contains a simple directed cycle $C'$. If there is none, then return that nothing was found. Otherwise, return the closed walk $C$ in $G$ to which $C'$ projects, and the subgraph $G_0$ of $G'$ to which $G'_0$ projects.

Let us now prove that this algorithm is correct. Every closed walk $C$ in $G$ whose image walk makes only $3_r$-turns is lifted by a directed cycle $C'$ in $G'$, and if $C$ is simple, then $C'$ is simple. Conversely, every simple directed cycle $C'$ in $G'$ projects to a closed walk $C$ that makes only $3_r$-turns. And if the connected component $G'_0$ of $C'$ contains no red vertex, then $C$ is simple; indeed, any self-intersection vertex of $C$ would either correspond to a self-intersection vertex of $C'$, which is impossible since $C'$ is simple, or otherwise it would correspond to red vertices in $C'$. We conclude with the claim that $G_0$ is the subgraph of $G$ spanned by the walks following $C$, and that $C$ is pulled left (resp. right) if and only if $G'_0$ contains a green (resp. red) vertex. Indeed, every walk $W$ that follows $C$ in $G$ lifts to a walk $W'$ based at some vertex of $C'$ in $G'_0$. Also, if $W$ can be extended by one edge into a walk that pulls $C$ to the left (resp. right), then the end-vertex $w'$ of $W'$ is marked green (resp. red). Reciprocally, any walk in $G'_0$ from $C'$ to a vertex $w'$ projects to a walk $W$ in $G$ that follows $C$. And if $w'$ is green (resp. red), then $W$ can be extended by one edge to pull $C$ to the left (resp. right).
\end{proof}

\begin{proof}[Proof of Theorem~\ref{harmonization theorem}]
  By \autoref{lem:simplicial}, we can assume that $f$ is simplicial.  Without loss of generality, we assume that $G$ is connected, for otherwise we could apply the algorithm to each connected component separately.  Making $f$ harmonious is done using balancings, shortenings, and flips.  Since balancings and shortenings decrease the length of the drawing, we give them priority over flips.  More precisely, the algorithm consists in applying the routine given below, which only performs flips, with the following important twist, left implicit in the description: \emph{whenever a balancing or a shortening is possible, we apply it and resume the routine from scratch}.  Recall that the simplicial map $f:G\to T^1$ factors as a homomorphism $\hat f:\hat G\to T^1$; it is convenient to express the routine in terms of~$\hat f$, since flips can be described at the homomorphism level.  Here is the routine:
  \begin{enumerate}
    \item Choose an arbitrary vertex~$r$ of~$\hat G$, and flip any vertex of $\hat G$ other than $r$, in any order, as long as possible.
    \item Direct $\hat G$ with the left-blue direction. If $\hat G$ has no directed cycle and has a single source, then do the following: apply \autoref{construct ordering} to build in $O(n)$ time a proper and monotonic ordering $v_0, \dots, v_{q-1}$ of the $q \geq 1$ vertices of $\hat G$; initialize $i := 0$.  Then, while it is possible to flip $v_{i\bmod q}$, flip it and increment~$i$.  (Some precisions: (a) we go to Step~3 as soon as $v_{i\bmod q}$ is not flippable; (b) The ordering is fixed during this entire step; if, after a flip, we update the direction of the edges to preserve the left-blue direction, it ceases to be monotonic.)
    \item Flip any vertex of $\hat G$ (even possibly $r$), in any order, as long as possible.
  \end{enumerate}
  
  If the algorithm terminates, then $f$ is harmonious by \autoref{no move implies harmonious}.  We now bound the number of flips of the routine, assuming that it is not interrupted by a balancing or a shortening. Step~1 does not flip $r$, so it consists of $O(n^2)$ flips by \autoref{not flipping the root}. Also, the flip sequence of Step~2 is proper, so it has length $O(mn)$ by Lemmas~\ref{flat zone}, \ref{straight implies flat zone}, and~\ref{we go forward}. Let $F$ be the flip sequence of Step~3.  For the sake of the analysis, we preserve the left-blue direction of the edges of~$\hat G$ after each flip (equivalently, at each flip of a vertex~$v$, we reverse the direction of the edges incident to~$v$).  We now prove that in all cases, some vertex of~$\hat G$ is not flipped in~$F$:
  \begin{itemize}
  \item If Step~2 was skipped because $\hat G$ has a directed cycle, then this cycle remains fixed by \autoref{oriented cycle that cannot be flipped};
  \item if Step~2 was skipped because $\hat G$ has a source $v$ distinct from $r$, then $v$ is a source that cannot be flipped, so $v$ is not flipped in Step~3 by \autoref{source that cannot be flipped};
  \item if Step~2 was executed, we claim that in Step~2, every attempt to flip a vertex $v$ happens when $v$ is a source.  Indeed, at the first round of flips, the neighbors of $v$ that have already been flipped correspond precisely to the edges that were directed towards $v$ in the initial monotonic ordering, which have thus been reversed by the flips, making $v$ a source.

  After each vertex of~$\hat G$ has been flipped, the directed graph is again monotonic.  This proves the claim.  It follows that $v$ cannot be flipped in Step~3 by \autoref{source that cannot be flipped}.
  \end{itemize}
  This implies that the flip sequence~$F$ of Step~3 has length $O(n^2)$ by \autoref{not flipping the root}, as desired.  Thus, overall, if not interrupted, the routine terminates after $O((m+n)n)$ flips if not interrupted.  Since there are $O(n)$ balancings or shortenings, the total number of flips, balancings, and shortenings is $O((m+n)n^2)$.

  To prove the claimed running time of $O((m+n)^2n^2)$ time, it remains to note that finding and applying the next move, or correctly asserting that no move can be applied anymore, takes $O(m+n)$ time. Indeed, $\hat G$ and $\hat f$ can be computed in $O(n)$ time by constructing the clusters of $f$. On $\hat G$ finding a flip, or correctly asserting that there is none, takes $O(m+n)$ time. Same for the shortenings. Concerning the balancings, an application of \autoref{find balancing} dermines in $O(m+n)$ time if there is a simple closed walk $C$ in $\hat G$ such that $\hat f \circ C$ makes only $3$-turns, is pulled left, and is not pulled right. If there is none, then no balancing is available. Otherwise, $C$ can be balanced, and \autoref{find balancing} computes both $C$ and the subgraph $\hat G_0$ of $\hat G$ that will move with $C$ during the balancing. The modification brought to $f$ by the balancing then takes $O(m+n)$ time.
\end{proof}


\section{Extensions to surfaces with boundary}\label{boundary}

In this section we extend the results of Sections~\ref{sec:tutte}~and~\ref{sec:harmonizing} to reducing triangulations with boundary. To mimic the classical setting of Tutte's theorem, we consider the constraint of attaching vertices to the boundary of the surface. We formalize this constraint in the following definitions. Let $S$ be a surface with non-empty boundary (such as the disk). Let $G$ be a graph, and let $f : G \to S$ be a map. We use the following conventions: the boundary $\partial S$ of $S$ is directed so that the interior of $S$ lies on its right, and the rotation system of a graph embedded on $S$ records, for every vertex $v$ embedded on $S$, the \emph{clockwise} order of the edges meeting $v$.

\paragraph{Anchor.} An \emphdef{anchor} is a (possibly empty) set $A$ of vertices of $G$ mapped to $\partial S$ by $f$, together with linear orderings of those vertices mapped to the same point of $\partial S$. Note that some vertices of $G$ can be mapped to $\partial S$ without belonging to $A$.

\paragraph{Untangling relatively to an anchor.} Let $f' : G \to S$ be obtained from $f$ by sliding infinitesimally along $\partial S$ the images of the vertices in $A$, so as to separate them in the orders prescribed by $A$. We say that $f$ can be \emphdef{untangled relatively to $A$} if there is an embedding homotopic to $f'$, where the homotopy fixes the image of every vertex in $A$.

\paragraph{Weak embeddings relative to an anchor.} Informally, we say that $f$ is a weak embedding \emph{relative to} $A$ if there exist embeddings arbitrarily close to $f$ in which the vertices in $A$ are embedded along $\partial S$, and in the orders prescribed by $A$. We shall not use this definition as is, and instead we consider an equivalent formulation in the case where $f$ is a drawing in the 1-skeleton $T^1$ of a reducing triangulation $T$ of $S$. This formulation better suits our needs. It is based on the combinatorial formulation of \autoref{sec weak embeddings}. More precisely, we extend $T^1$ to a graph $M^*$ equipped with a rotation system, we extend $G$ and $f$ to a graph $G^*$ and a drawing $f^* : G^* \to M^*$ as follows. Consider every two consecutive edges $e$ and $f$ along the boundary of $T$, and let $x$ be the vertex of $T$ between $e$ and $f$. Let $v_1, \dots, v_k$ be the $k \geq 0$ vertices in $A$ that are mapped to $x$ by $f$, in the order prescribed by $A$. In $M^*$, create $k$ degree one vertices attached to $x$ by edges $s_1, \dots, s_k$: place them in the rotation system of $M^*$ so that $e, s_1, \dots, s_k, f$ are consecutive around $x$. Then for every $1 \leq i \leq k$, create an edge incident to $v_i$ in $G^*$, and map this edge to $s_i$ in $f^*$. We say that $f$ is a \emphdef{weak embedding relative to} $A$ if $f^*$ is a weak embedding.

Our result is the following:

\begin{theorem}\label{tutte with boundary}
Let $S$ be a surface with boundary. Let $T$ be a reducing triangulation of $S$, with $m$ edges. Let $G$ be a graph of size $n$, and let $f : G \to T^1$ be simplicial. Let $A$ be an anchor for $f$. If $f$ can be untangled relative to $A$, then we can compute in $O((m+n)^2 n^2)$ time a simplicial map $f': G \to T^1$, homotopic to $f$ relatively to $A$, weak embedding relative to $A$, not longer than $f$.
\end{theorem}

\begin{figure}
    \centering
    \includegraphics{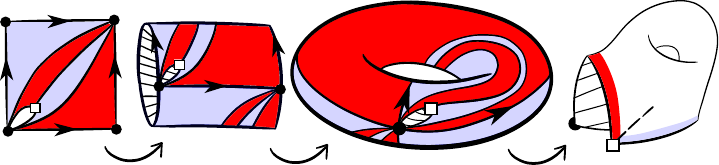}
    \caption{Construction of the $1$-gadget.}
    \label{fig:gadget}
\end{figure}

\begin{cfigure}
    \centering
    \includegraphics{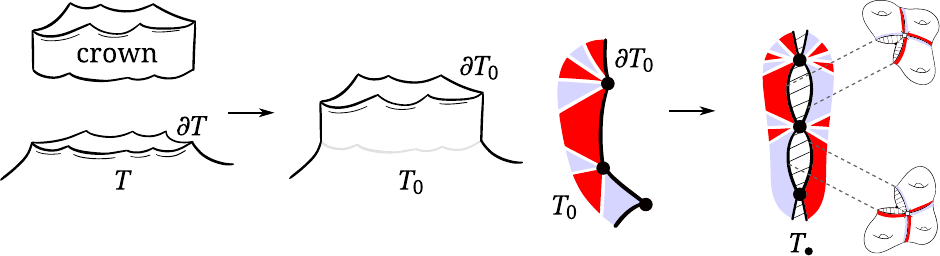}
    \caption{Construction of the triangulation $T_\bullet$ in the proof of \autoref{tutte with boundary}.}
    \label{fig:extension}
\end{cfigure}

The proof of \autoref{tutte with boundary} goes by extending $T$ to a reducing triangulation without boundary (in order to apply \autoref{harmonization theorem}). This extension step uses ad-hoc gadgets that we now define. Let $S$ be the surface of genus one with one boundary component. The \emphdef{1-gadget} is the reducing triangulation $T$ of $S$ depicted in \autoref{fig:gadget}. The boundary of $S$ is the union of two edges of $T$: one edge is incident to a blue face of $T$, and the other to a red face. The \emphdef{$3$-gadget} is then the reducing triangulation $T'$ obtained from disjoint 1-gadgets $T_1, T_2, T_3$ by identifying the blue-incident edge of $T_i$ with the red-incident edge of $T_{i+1}$ for every $i \in \{1,2\}$. The boundary of $T'$ is the union of the red-incident edge of $T_1$ and the blue-incident edge of $T_3$.

We call \emphdef{crown} any reducing triangulation $T$ of the annulus obtained from a circular list of $k \geq 2$ triangles $t_0, \dots, t_{k-1}$ by identifying some side of $t_i$ with some side of $t_{i+1}$ for every $i$, where indices are modulo $k$. 

\begin{proof}[Proof of Theorem~\ref{tutte with boundary}]
We shall extend $T$ to a reducing triangulation \emph{without boundary} $T_\bullet$, and then extend $f$ to a drawing on $T_\bullet$. The construction of $T_\bullet$ is as follows. First extend $T$ to a reducing triangulation $T_0$ by attaching a crown to each boundary component of $T$. Choose the crowns so that if $x$ is a vertex of the boundary of $T$, and if $k \geq 0$ vertices of $A$ are mapped to $x$ by $f$, then $x$ is incident to at least $k+6$ edges in the interior of its crown; Crucially, we make the observation (O) that the first three, and the last three edges incident to $x$ in the interior of the crown will not be used by the extended drawing. Now build $T_\bullet$ from $T_0$ as follows. See \autoref{fig:extension}. Consider a copy $T_1$ of $T_0$. Reverse the direction of $T_1$, and exchange the colors of its faces. Identify every edge of the boundary of $T_0$ with its copy in $T_1$. At this point the faces of $T_\bullet$ are properly colored, but $T_\bullet$ may not be a reducing triangulation since vertices on the boundary of $T_0$ may have low degree. This is fixed by cutting open in $T_\bullet$ every edge of the boundary of $T_0$, and by identifying the two cut paths with the two boundary paths of a 3-gadget (keeping the coloring of the faces proper).

Now extend $G$ and $f$ to a graph $G_\bullet$, and a simplicial map $f_\bullet : G_\bullet \to T^1_\bullet$; First extend $G$ and $f$ to some $G_0$ and $f_0 : G_0 \to T^1_0$, as follows. Consider every vertex $x$ of the boundary of $T$. Let $v_1, \dots, v_k$ for some $k \geq 0$ be the vertices of $A$ mapped to $x$ by $f$, in order. Let $e_1, \dots, e_{k+6}$ be consecutive edges around $x$ in the interior of the crown incident to $x$. For every $1 \leq i \leq k$, add in $G_0$ an edge $s_i$ incident to $v_i$, and let $f_0(s_i) := e_{i+3}$. We say that the newly created end-vertex of $s_i$ is a tip vertex of $G_0$, and that $s_i$ is a tip edge of $G$. Build $G_\bullet$ and $f_\bullet$ from $G_0$ and $f_0$ by considering a copy $G_1$ of $G_0$, and the mirror map $f_1 : G_1 \to T_1^1$, and by identifying in $G_0 \cup G_1$ and $f_0 \cup f_1$ every tip vertex of $G_0$ with its mirror vertex in $G_1$.

To conclude, transform $f_\bullet$ to a map $f'_{\bullet} : G_\bullet \to T_\bullet^1$ with the algorithm of \autoref{harmonization theorem}, and return $f' := f'_\bullet|_G$. Let us show why this is correct. Let $T^\flat$ be the sub-triangulation of $T_\bullet$ that is the union of $T$ and the mirror of $T$. Let $G^\flat$ be the subgraph of $G_\bullet$ that is the union of the two copies of $G$. By the above observation (O), applying the moves described in \autoref{sec:harmonizing} to $f_\bullet$ does not modify the images of the stem edges, and preserves the fact that $f_\bullet(G^\flat) \subset T^\flat$. In particular $f'$ is homotopic to $f$ relative to $A$. Moreover $f_\bullet$ is homotopic to an embedding in $T_\bullet$ since $f$ can be untangled relatively to $A$ in $T$. Thus, $f_\bullet'$ is a weak embedding by \autoref{tutte theorem}, and $f_\bullet'$ is not longer than $f_\bullet$. And so $f'$ is a weak embedding relative to $A$, not longer than $f$.
\end{proof}

\subsection*{Acknowledgments}

We thank the reviewers of this paper and of its conference version for their detailed comments that led to several improvements in the presentation and to a critical fix in the definition of a swap.

\printbibliography

\pagebreak
\appendix

\section{Proof of Lemmas~\ref{torus problem} and~\ref{higher genus problem}}\label{problem proof}

In this section we prove Lemmas~\ref{torus problem} and~\ref{higher genus problem}, which we restate for convenience.

\torusproblem*
\begin{figure}
    \centering
    \includegraphics[scale=0.4]{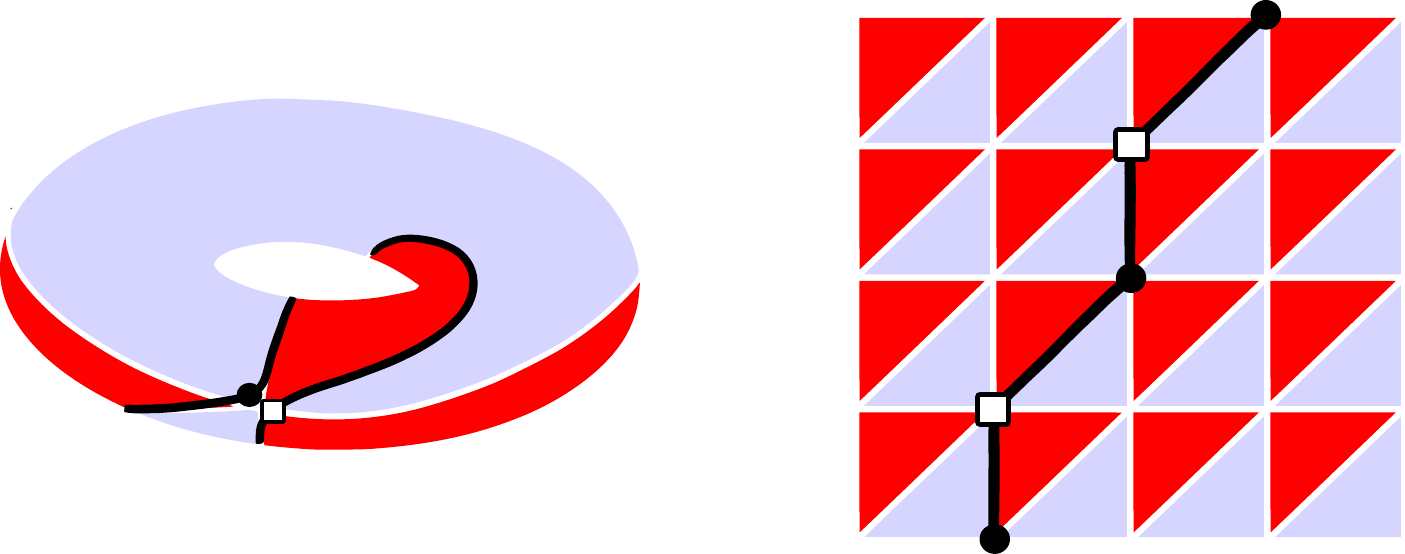}
    \caption{(Left) A reducing triangulation $T$ of the torus, and a closed walk $C$ of length two in $T$. (Right) A portion of the universal covering triangulation of $T$, and a portion of a lift of $C$. In $T$ there is no reduced closed walk freely homotopic to $C$. }
    \label{fig:torus}
\end{figure}

\begin{proof}
Consider the reducing triangulation $T$, and the closed walk $C$ of \autoref{fig:torus}. Assume the existence of a reduced closed walk $C'$ in $T^1$. If at some point $C'$ takes a directed edge of $T$ that sees red on its left, then $C'$ makes only $3_r$-turns. Otherwise, $C'$ makes only $3_b$-turns. In both cases $C'$ is not freely homotopic to $C$.
\end{proof}

\highergenusproblem*

\begin{proof}
There is a reducing triangulation $T$ of $S$ whose vertex degrees are all greater than or equal to eight, see~\cite[Figure~17]{de2024untangling}. There is a closed walk $C_2$ in $T^1$ that makes only $3_r$-turns; indeed every walk that makes only $3_r$-turns will repeat itself since $T$ is finite. Let $C_1$ be the closed walk obtained by pushing $C_2$ to the left in the way depicted in~\cite[Figure~9]{de2024untangling}: $C_2$ makes only $-3_r$-turns. Let $C'$ be a closed walk in $T^1$, freely homotopic to $C_1$ or $C_2$. Assume by contradiction that $C'$ is strongly harmonious. Then $C'$ is reduced, and it follows from the work of É.~Colin de Verdière, Despré, and Dubois~\cite{de2024untangling} that $C'$ is equal to $C_1$ or to $C_2$ (up to cyclic permutation and reversal). Indeed, $T$ fits their more restrictive definition of reducing triangulation. Also, if $C'$ does not make only $3_r$-turns, then $C'$ fits their more restrictive definition of reduced closed walk, and so $C' = C_1$ by~\cite[Proposition~3.4]{de2024untangling}. If $C'$ makes only $3_r$-turns, then reversing the colors of $T$ makes $C'$ reduced under their definition, and gives $C' = C_2$. So either $C_1$ or $C_2$ is strongly harmonious, which is not the case, a contradiction.
\end{proof}

\end{document}